\documentclass[letterpaper,11pt]{article}
\usepackage[utf8]{inputenc}

\pdfoutput=1

\usepackage{amsmath, amsthm, amssymb}

\usepackage{algpseudocode,algorithm}
\usepackage{algorithmicx}
\usepackage{mathtools}
\usepackage[numbers]{natbib}
\usepackage{comment} 
\usepackage{color-edits} 
\usepackage[most]{tcolorbox}
\usepackage{xfrac}
\usepackage{hyperref}
\usepackage{multirow}
\usepackage{caption}
\usepackage{bm}
\usepackage{newfloat}
\usepackage{scalerel}
\usepackage{enumitem}
\usepackage{accents}
\usepackage{xpatch}
\usepackage{slashbox}

\makeatletter
\xpatchcmd{\algorithmic}
  {\ALG@tlm\z@}{\leftmargin\z@\ALG@tlm\z@}
  {}{}
\makeatother

\def \a {{{\sf{a}}}}
\def \qq {{\theta}}

\def \bu {{{\sf{b}}}}
\def \A {{{\sf{A}}}}
\def \B{{{\sf{B}}}}

\def \esp{\text{\sf{ESP}}}

\def\res{{{\text{\sf {R}}}}}

\def\q{q}

\def\rlp{\b{r}^{\scriptscriptstyle{\text{out}}}}
\def\rout{\b{r}^{\scriptscriptstyle{\text{out}}}}
\def\rr{\b{r}^{\scriptscriptstyle{\mathcal R}}}

\def\rrb{\b{r}^{\scriptscriptstyle{\mathcal R}}_{\bu}}

\def\name{{\sf{\small{Pro-LPR}}}}
\usepackage{xcolor}

\renewcommand{\comment}[1]{}
\newcommand{\Halmos}[0]{}

\newtheorem{manualtheoreminner}{Theorem}
\newenvironment{manualtheorem}[1]{%
  \renewcommand\themanualtheoreminner{#1}%
  \manualtheoreminner
}{\endmanualtheoreminner}

\usepackage[margin=1in]{geometry}
\allowdisplaybreaks

\definecolor{mygreen}{RGB}{20,120,60}

\hypersetup{
     colorlinks=true,
     citecolor= mygreen,
     linkcolor= mygreen
}

\title{LP-based Approximation for Personalized \\ Reserve Prices \footnotetext{\hspace{-0.2 mm}An earlier version of the paper has appeared at the ACM Conference on Economics and Computation (EC), 2019.}}

\author{Mahsa Derakhshan\thanks{ Research conducted in part while the author was an intern at Google.} \\ University of Maryland 
\and Negin Golrezaei  \\  MIT\and
Renato Paes Leme \\ 
Google Research
}

\date{}

\definecolor{mygreen}{RGB}{20,140,80}
\definecolor{mylightgray}{RGB}{230,230,230}
\newcounter{myalgctr}

\newenvironment{tbox}{
\vspace{0.2cm}
\begin{tcolorbox}[width=\textwidth,
                  enhanced,
                  boxsep=2pt,
                  left=1pt,
                  right=1pt,
                  top=4pt,
                  boxrule=1pt,
                  arc=0pt,
                  colback=white,
                  colframe=black,
                  breakable
                  ]
}{
\end{tcolorbox}
}

\newcommand{\tboxhrule}[0]{\vspace{0.1cm} \hrule \vspace{0.2cm}}

\renewcommand{\b}[1]{\ensuremath{\bm{\mathrm{#1}}}}

\DeclareMathOperator*{\E}{\mathbb{E}}

\DeclareFloatingEnvironment[fileext=frm,name=LP]{lp}

\newcounter{eqcounter}
\addauthor{md}{blue}  

\newtheorem{theorem}{Theorem}
\newtheorem*{theorem*}{Theorem}

\newtheorem{lemma}{Lemma}[section]

\newtheorem{corollary}[lemma]{Corollary}

\newtheorem{definition}[lemma]{Definition}

\definecolor{mygreen}{RGB}{20,100,60}

\hypersetup{
     colorlinks=true,
     citecolor= mygreen,
     linkcolor= black
}

\algnewcommand{\IIf}[1]{\State\algorithmicif\ #1\ \algorithmicthen}
\algnewcommand{\EndIIf}{\unskip\ \algorithmicend\ \algorithmicif}

\newcommand{\apxf}[0]{\ensuremath{0.684}}

\newcommand{\rev}[2]{\ensuremath{\text{\sf{Rev}}}_{#1}(#2)}

\newcommand{\lps}[0]{\ensuremath{s^{\star}}}

\newcommand{\revision}[1]{{\textcolor{black}{#1}}}

\newenvironment{titledtbox}[1]{\begin{tbox}#1
\tboxhrule
}{\end{tbox}}

\newcommand{\abs}[1]{\left\vert{#1}\right\vert}

\newcommand{\Rev}{\operatorname{Rev}}
\renewcommand{\Pr}{\operatorname{Pr}}
\newcommand{\optf}[2]{\ensuremath{\text{\sf OPT}(#1)}}
\newcommand{\optfone}[2]{\ensuremath{\text{\sf OPT}}^1(#1, #2)}
\newcommand{\optftwo}[2]{\ensuremath{\text{\sf OPT}}^2(#1, #2)}

\begin{document}
\maketitle

\begin{abstract}
  We study the problem of computing data-driven personalized
 reserve prices in eager second price auctions 
 without having any assumption on
  valuation distributions.
Here, the input is a data-set that contains the submitted bids of $n$ buyers
in a set of  auctions and the problem  is to return personalized reserve prices $\b
r$ that maximize the  revenue earned on these auctions by running  eager
  second price auctions with reserve $\b r$. {For this problem, which is known to be NP-hard, we present a novel LP
  formulation and a rounding procedure which achieves}
  a $(1+2(\sqrt{2}-1)e^{\sqrt{2}-2})^{-1} \approx \apxf$-approximation. This improves over the  $\frac{1}{2}$-approximation
algorithm due to Roughgarden and Wang. {We show that our analysis is
  tight for this rounding procedure. We also bound the integrality gap of the LP, which shows that it is impossible to design an algorithm with an approximation factor larger than $0.828$ with respect to this LP. 
  }
\end{abstract}

\section{Introduction}

{\em Second price (Vickrey) auctions with reserves} have been prevalent in many marketplaces such as online advertising markets \citep{chawla2014bayesian, paes2016field, golrezaei2017boosted}. A key parameter of this auction format is its {\em reserve price}, which is the minimum price at which the seller is willing to sell an item. While there is empirical and theoretical evidence that highlights the significance of setting {\em personalized reserve prices} for the buyers to maximize the revenue \citep{ edelman2007internet, ostrovsky2011reserve,beyhaghi2018improved}, we do not have a full understanding of how to optimize reserve prices. This problem is  only  fully solved  under the assumption that  buyers' valuation distributions are i.i.d. and regular, where  these  assumptions fail to hold in practice 
\citep{ celis2014buy,golrezaei2017boosted}.

We study the problem of optimizing personalized  reserve prices in 
second price auctions  when the buyer valuations can be correlated.   There are two different ways that personalized 
	reserve prices can be applied in  second price auctions: lazy and eager
	\cite{dhangwatnotai2015revenue}. In the lazy version,  we first determine the potential
	winner and  then apply the reserve prices. In the eager version,  we first
	apply the reserve prices and  then determine the winner. In this work, we focus on optimizing  eager reserve prices because  (i) while the optimal lazy
	reserve prices can be computed exactly  in polynomial time, they have worse
	revenue performance both in theory and practice, and (ii) eager reserves  perform better in
	terms of social efficiency for similar revenue levels \citep{paes2016field}.

To optimize the eager reserve prices, we take a data-driven approach as suggested in the literature  \citep{paes2016field, RW16}.  The input in this setting is a history of the buyers' submitted bids/valuations over multiple runs of an auction and the goal, roughly speaking, is to set a personalized reserve price $r_{\bu}$ for each buyer $\bu$ such that the total revenue obtained on the same data-set according to these reserve prices is maximized (see Section~\ref{sec:pre} for the formal definition).

\revision{The optimal  data-driven reserve prices solve an \emph{offline
optimization problem}, i.e., given a data-set of bid data, it computes the optimal 
reserve prices in retrospect. 
 {This approach is mainly inspired  by online advertising markets,  in which billions of second price auctions are run in a day by ad exchanges. This practice   provides ad exchanges a big data-set of submitted bids, which can be used to optimize personalized reserve prices for the next day via a data-driven approach. Setting personalized reserve prices, which is a common practice in this market, is inspired by targeting tactics in which advertisers take advantage of cookie-matching technologies to target Internet users based on their diverse  preferences. This then leads to heterogeneous valuation/bid distributions, which necessitates setting  personalized reserve prices.} }

\textbf{Prior Work and Our Results:} Unfortunately,  optimizing the data-driven reserve prices is APX-hard \citep{RW16}. The state-of-the-art algorithm of    \cite{RW16} achieves a $1/2$-approximation which itself improves over an earlier $1/4$-approximation algorithm by \cite{paes2016field}. 
Our main result is an algorithm with a significantly improved approximation factor. We show that there exists a randomized polynomial time algorithm that given a data-set, outputs a vector of reserve prices whose expected revenue is a \apxf{}-approximation of that of the optimal value. Our improved bound  results in a polynomial time $(\apxf -\epsilon)$-algorithm for independent distributions, which beats the best approximation known via prophet techniques.\footnote{While
we provide a better guarantee against the optimal reserves, our technique
does not provide approximation guarantees with respect to the optimal auction as
prophet inequalities do.}
Further, our result leads to a
$(\apxf - \epsilon)$-algorithm for the batch learning version of the problem.   
 \revision{In the batch learning setting,  
 there is a distribution over buyers'  valuations/bids  and the goal is to   compute the optimal prices  by  having  access to samples  from that distribution 
\citep{medina2014learning,huang2018making}. 
Using the machinery developed by 
\cite{morgenstern2015pseudo}, one can show that  via solving the data-driven offline
optimization problem on the data-set with $\Omega(\vert \B \vert \log \vert \B
\vert /\epsilon^2)$ auctions, we can obtain a $1-\epsilon$ fraction of the maximum
revenue of any eager second price auction
that one could have hoped to obtain by knowing the valuation
distributions.}

The known algorithms of the literature are all greedy and only take into account
the two highest bids in each auction. Another limitation of these  algorithms is
that the reserve price for each buyer is computed in isolation. That is,  the
reserve price for a buyer only depends on the bids of the auctions in which  the
buyer submits the highest bid. In fact, \cite{RW16} argue that these
limitations are precisely what prevent their algorithm from obtaining any
guarantee better than $1/2$. As we explain in more detail later, we bypass this bound by a careful  analysis of a rounding technique  for a natural linear programming formulation of the problem proposed in this work.

{\textbf{Our Techniques:}}  
  To obtain our improved approximation factor of \apxf{}, 
     we present an algorithm called ``\textbf{Pro}file-based
     \textbf{LP}-\textbf{R}ounding", \name{} for short,  that takes advantage of
     a concise representation of the solution space. This  representation, that
     we call profile space, is inspired by how revenue is computed in the eager
     auctions. Working with the profile space enables us to consider all the
     bids in
      an auction---not only the highest and second highest bids---to set the
      reserve prices. It further  allows  us to describe the optimal solution by a polynomial-size integer
      program.
	   	   By relaxing the integrality constraints on the variables of the integer program, we construct a linear program (LP). The fractional solution of the LP is then rounded to obtain the reserve prices.  
  The final reserve price of the algorithm is the best of the zero reserves and
  the reserves obtained from rounding the solution of the LP.  {The
  most technically challenging step in the analysis is to bound the
  approximation ratio. This is done via careful probabilistic analysis of the
  rounding procedure which leads to a non-linear mathematical program bounding
  the ratio. Our last step is to use techniques from non-linear optimization to
  bound the solution of the mathematical program.}  {We would like to emphasize that our analysis of our algorithm is tight  in a sense that there is an example for which our algorithm cannot get  an approximation factor better than   \apxf.}
 
{We point out  that  the performance of our algorithm is evaluated against the optimal value of the LP, which is an upper bound on the maximum revenue. By analyzing the integrality gap of the LP, we show that no algorithm can obtain more than a 0.828 fraction of the optimal value of the LP; see Theorem \ref{thm:integrality_gap}. This highlights that  our algorithm is evaluated  against a powerful benchmark and despite that,  it obtains \apxf~ fraction of this powerful benchmark.}

\textbf{Managerial Insights and Numerical Studies:} 
{Our proposed algorithm highlights  that there is significant value in considering all the submitted bids, not only the highest and second highest bids, to optimize reserve prices.  By considering all the submitted bids, the algorithm can better identify and take advantage of  the buyers' bidding behavior. Furthermore, the design of our algorithm accentuates the importance of  optimizing  reserve prices jointly. (Recall that the prior work focused on optimizing reserve of each buyer separately.)  Such  a joint optimization problem can capture potential correlation in the submitted bids, and hence improve revenue. 
{To illustrate this, we conduct numerical studies, where we compare our algorithm with the greedy algorithm of \cite{RW16}. As stated earlier, this greedy algorithm obtains the best approximation factor prior to our work. We show that when bids are positively or negatively correlated across buyers, (i) our algorithm  obtains at least a $0.98$ fraction of the optimal revenue, and (ii) in $50\%$ (respectively $75\%$) of the problem instances, 
 our algorithm outperforms the greedy algorithm by at least $6\%-7\%$ (respectively $2\%-3\%$). We obtain similar results  when bids are independent across buyers.} }  
 
{\subsection{ {Other Related Work}} \label{sec:related}
\revision{While we have discussed several closely related works, we further situate our work in the landscape
of related work. 
 We start with a more detailed comparison between our work and that of \cite{RW16}. 
 As stated earlier, the greedy nature of \cite{RW16}'s algorithm prevents it from obtaining an approximation factor better than $1/2$. However, the greedy nature of 
 their algorithm allows them to transform it into an online learning  algorithm with a sublinear approximate regret. The online learning algorithm receives a set of submitted bids to an auction and its goal to set personalized reserve prices based on collected feedback in past auctions.  \cite{RW16} design a learning algorithm under a full information setting, where the auctioneer observes all the bids after running an auction. Very recently, using Blackwell Approachability,
  \cite{niazadeh2020online} show how to transform a variation of the greedy algorithm of  \cite{RW16}  into its online counterpart under a bandit feedback structure. In this structure, after every auction, the auctioneer only observes the obtained revenue, rather than all the submitted bids. In this work, we only focus on the offline problem of optimizing reserve prices. Considering our improved approximation factor,  an interesting future research direction is to study how to transform the \name{} algorithm to its online counterpart with sublinear approximate regret.}

{Our work relates and contributes to the literature on the auctions and revenue-maximizing mechanisms in a single-item environment. Seminal contributions have been made by in  \cite{myerson1981optimal} under a critical  assumption that  buyers' valuations are independent of each other. Specifically, \cite{myerson1981optimal} shows that
  
  when buyers' valuation distributions are regular and i.i.d., the optimal mechanism can be implemented by running second price auctions with a non-anonymous reserve price.  However, 
even under the independence assumption, when valuation distributions are
irregular or heterogeneous, the optimal mechanism  has a rather complicated
structure and its implementation heavily relies on the knowledge of the
valuation distributions \citep{celis2014buy, roughgarden2016ironing, golrezaei2017boosted}.}

  {Considering the complexity of the optimal auction, there has been a growing body of  literature on designing simple
  yet effective auctions that can be easily optimized; see, for example,
  \cite{golrezaei2017boosted, celis2014buy, paes2016field, RW16,
  allouah2018prior, bhalgat2012online, beyhaghi2018improved,
  dhangwatnotai2015revenue}. Among those is second price auctions with reserve, which is very common in
practice \citep{beyhaghi2018improved, paes2016field, chawla2014bayesian}. This
auction format has a single parameter, called reserve price, per buyer that
determines  the minimum acceptable bid for the buyer. As stated earlier, to
extract high revenue from buyers, it is very crucial to effectively optimize
reserve prices. Such an optimization problem has been studied in different
settings. } 

\revision{If the value distributions are independent, an improved approximation
to personalized reserves are known via techniques like the correlation gap
\citep{chawla2010multi, yan2011mechanism} and prophet inequalities
\citep{KS78,HK81,azar2017prophet,esfandiari2017prophet,beyhaghi2018improved,correa2019prophet}
(to cite a few). The latest result is $0.669$-approximation by 
\cite{correa2019prophet}. Although those results are typically stated as an
approximation ratio with respect to the (stronger) Myerson revenue benchmark,
those are also the best-known approximation ratios with respect to the optimal
reserve prices for independent distributions.}

Another related stream of literature is on the design of auctions for
correlated distributions. This line of work was pioneered by  \cite{ronen2001approximating} and
\cite{ronen2002hardness}. The positive and negative results
were later improved by \cite{dobzinski2011optimal} and
\cite{papadimitriou2011optimal}. Our paper departs
from this line work in the sense that we do not try to approximate the optimal
incentive-compatible auction, but instead, we try to approximate the best auction in
the subclass of second price auctions with reserve prices, since this 
is the auction format adopted by most online marketplaces, including online display advertising markets.} \revision{To optimize reserve prices in second price auctions, as stated earlier, we take a data-driven approach. This approach allows us to
indirectly exploit potential correlations in the submitted bids. In online advertising, for example, such correlation exists when different buyers/advertisers value some features of ad impressions, in the same way, e.g., they all prefer showing their ads to users who tend to click more on ads. 
}

{Finally, we review some of the work that bears some resemblance to our work from a technical perspective. 
In our work, we change 
 the solution space to describe  the optimal solution  using a concise LP with a small integrity gap. A similar technique is used in different applications; see, for example, \cite{behnezhad2017polynomial, behnezhad2017faster, behnezhad2018spatio, behnezhad2019optimal}, \cite {ahmadinejad2019duels} and \cite{immorlica2011dueling}
 for the use of a similar technique in finding optimal strategies of Blotto, security games, and dueling games. 
  While sharing the general technique, each paper uses specific properties of their problem to design their alternative solution space.  Furthermore, 
there is another related line of work that design ``configuration" LPs for problems related to resource allocation and job scheduling,  \citep{polavcek2015quasi, svensson2012santa, asadpour2010approximation, bansal2019lift}. This line of work is initiated by  \cite{bansal2006santa}. We note that configuration LPs, unlike our LP, may not be polynomial in the input size.
}

{The rest of the paper is organized as follows. In Section \ref{sec:pre}, we
define the model. Section \ref{sec:results} presents a high-level view of the
results and techniques. In Section \ref{section:lp}, we provide our LP, which
will be used as our benchmark. In Section \ref{section:rounding}, we present the
\name{} algorithm and show its performance guarantee. Section
\ref{sec:gap} provides the proof of the integrality gap and Section
\ref{sec:tight} shows that our analysis is tight. Finally, we present the results of our numerical studies in Section~\ref{sec:numerical} and conclude in Section~\ref{sec:conclusion}.} 

\section{Preliminaries and Problem Statement} \label{sec:pre}
 There are $n$ buyers participating in a set of single-item eager second price auctions. Let $\A$ and $\B$ respectively denote the set of auctions and buyers. For any buyer $\bu\in \B$, and for any auction $\a\in \A$,  we are given
a non-negative number $\beta_{\a, \bu}$ 
which indicates the bid of buyer $\bu$ in auction $\a$. 
Let $r_{\bu}$ be the personalized reserve price of buyer $\bu\in \B$. Then, given
the bids $\{\beta_{\a, \bu}\}_{\bu \in \B}$ in auction $\a\in \A$ and  reserve
prices $\b r =\{r_{\bu}\}_{\bu\in \B}$, the eager second price ($\esp$) auction
works as follows. 

\hspace{-3mm} - First, any buyer $\bu$ with $\beta_{\a, \bu} < r_{\bu}$ is
  eliminated. Let $S_{\a}=\{\bu: \beta_{\a, \bu}\ge r_{\bu}\}$ be the set of buyers who clear their reserve prices in auction $\a$.
  
\hspace{-3mm} - When set $S_{\a}$ is nonempty, the item is  allocated to buyer
$\bu_{\a}^{\star} =\arg\max_{\bu\in S_{\a}}~\{\beta_{\a, \bu}\}$ who has the
    highest bid among all the buyers in set $S_{\a}$ and is charged
    $${\rev{\a}{\b r}}
:=\max\left\{r_{\bu_{\a}^{\star}},~ \max_{\bu\in S_{\a}, \bu\ne
    \bu_{\a}^{\star}}~\{\beta_{\a, \bu}\} \right\}\,.$$ {Note that $S_{\a}$ and
    $\bu_{\a}^{\star}$ implicitly depend on reserve prices $\b r$.} Any other buyer $\bu\in \B$,
$\bu\ne \bu_{\a}^{\star}$ is not charged. Further, when set $S_{\a}$ is empty,
the item is not allocated and ${\rev{\a}{\b r}} = 0$.

Note that the reserve prices are the same across all the auctions $\a \in \A$. However, each buyer $\bu$ is assigned a personalized reserve price $r_{\bu}$. Given the data-set of bids $\{\beta_{\a, \bu}\}_{\a\in \A, \bu\in \B}$, our goal here is to find personalized reserve prices that maximize revenue of the auctioneer. {See the introduction section for a discussion on the nice properties of this data-driven optimization.} 
 Formally, we would like to solve the following optimization problem:
 \begin{align}
   \esp^{\star}~=~\max_{\b r \in \mathbb{R}^{n}}~~ \rev{}{\b r} := \sum_{\a \in
   \A}\rev{\a}{\b r}\,.
\label{eq:opt} \tag{\small{\sf ESP-OPT}}
 \end{align} 
 Note that, without loss of generality, we  
 assume that the optimal reserve price for buyer $\bu$ is equal to one of his submitted bids $\{\beta_{\a, \bu}\}_{\a\in \A}$. Let $\res = \{0, \infty\} \cup\{\beta_{\a, \bu}\}_{\a\in \A, \bu\in \B}$. Then, Problem \ref{eq:opt} can be rewritten as 
$ \max_{\b r \in {\res^{n}}}~\sum_{\a \in \A}
{\rev{\a}{\b
r}}$, which leads to a search space of  size {$\abs{\res}^{n}$}. 
   
   {We now make a few remarks about the    data-driven optimization Problem \eqref{eq:opt}.}
   {The optimal solution to the   data-driven optimization Problem \eqref{eq:opt} gives  an approximation solution to the batch-learning setting. In this setting,   bids in each auction $\a$, i.e., $\{\beta_{\a, \bu}\}_{\bu\in \B}$, are   independent samples from the distribution
 $\mathcal{D}$.
By analyzing  the sample complexity of the auctions,  \cite{medina2014learning} and \cite{morgenstern2015pseudo} show that in the batch learning setting, 
 with probability $1-\delta$, it holds that:
$$\left\vert \E_{\boldsymbol{\beta} \sim \mathcal{D}}[\Rev(\boldsymbol{\beta}; \b r)] -  \frac{1}{m} \sum_{\a=1}^m \Rev(\{\beta_{\a, \bu}\}_{\bu\in \B}; \b r) \right\vert \leq O \left( \sqrt{\frac{n \log(m/\delta)}{m}} \right), $$
where with a slight abuse of notation, $\Rev(\boldsymbol{\beta}; \b r)$ is the revenue of an auction when bids and reserve prices are $\boldsymbol{\beta}$ and $\b r$, respectively. This implies that by having   $m = \Omega(n \log n / \epsilon^2)$ samples, with high probability, the optimal solution to  the   data-driven optimization Problem \eqref{eq:opt} provides  an $\epsilon$-additive approximation solution to the batch-learning problem when the data-set of bids is drawn from a distribution $\mathcal D$. 
}

{Motivating the data-driven problems via the   batch-learning problems implies that the bid distribution is the same across all the auctions. The consistency of the bid distribution across auctions, which is a common assumption  in the literature (e.g., \cite{paes2016field, RW16}),  can be justified  
when buyers
  submit their true valuations in auctions, i.e., when they are truthful. The truthfulness is, in fact, an appealing property of  (single-shot) second price auctions. In single-shot second price auctions,  buyers are only concerned about their utility in the current auctions and do not reason about  how their bids will affect future outcomes. In other words,  buyers are myopic, rather than being forward-looking/strategic.   Nonetheless, there is a new line of work that studies how to optimize reserve prices in (lazy) second price auctions when buyers are forward-looking and would like to maximize their cumulative utility; see, for example,  \cite{amin2013learning, amin2014repeated, kanoria2019incentive}, and \cite{golrezaei2018dynamic, golrezaei2019incentive}. At a high level, it has been shown that in lazy second price auctions, it is possible to effectively optimize reserve prices even when buyers are forward-looking. We believe that some of the techniques developed for lazy second price auctions can be applied to eager second price auctions as well. Investigating this  claim  is indeed an interesting future research direction.     
 }

\section{Results and Techniques}\label{sec:results}
The main result of the paper is a randomized algorithm that returns an \apxf{}-approximation solution for Problem \ref{eq:opt}.

\begin{theorem}[Main Theorem]\label{thm:main}
There exists a randomized polynomial time algorithm that given a data-set $\{\beta_{\a, \bu}\}_{\a\in \A, \bu\in \B}$, outputs a vector of eager reserve prices whose expected revenue is a \apxf{}-approximation of that of the optimal value  of Problem  \ref{eq:opt}, denoted by $\esp^{\star}$.
\end{theorem}

To find an approximate solution, the overall idea is to  construct an LP whose objective function at its optimal solution provides an upper bound for $\esp^{\star}$. The LP that takes advantage of a concise representation of the solution space, has a polynomial number  of variables and constraints. 
Then, we use a rounding technique to transform the optimal solution of the LP to
a vector of reserve prices. We show that if we consider the reserve prices
obtained from the rounding technique  and the vector of all-zero reserve prices
and choose the one with the maximum revenue, we obtain the desired
approximation factor. {In Theorem \ref{thm:tight}, we further show that our
analysis of our approximation factor is tight. That is, we provide an example
for which our algorithm obtains exactly   \apxf{} fraction of the optimal value
of the LP, i.e., the upper bound on for $\esp^{\star}$.
Finally, in Theorem \ref{thm:integrality_gap}, we bound the integrality
gap of the LP. This characterization shows that no algorithm can obtain more
than $0.828$ fraction of the LP.}

\section{Linear Program} \label{section:lp}
{The main challenge in designing an LP formulation for this problem is to
find a concise representation of the solution space. Instead of considering all
possible assignments of reserves to buyers, we will consider only partial
assignments in which we only specify the reserve prices of two buyers. We will
call such partial assignment a \emph{profile}. Formally, a profile is a tuple $p =
(\bu_1, \bu_2, r_1, r_2) \in \B \times \B \times \res \times \res$, which
represents an assignment of reserve $r_1$ to buyer $\bu_1$ and reserve $r_2$ to
buyer $\bu_2$. If it is the case that the reserves are below the corresponding
bids in an auction $\a$, i.e. $r_1 \leq \beta_{\a,\bu_1}$ and $r_2 \leq
\beta_{\a,\bu_2}$, then no matter how the assignment of the remaining reserves,
the revenue of this partial assignment is at least $\max\{ r_1,
\beta_{\a,\bu_2}\}$ for $\beta_{\a,\bu_1} \geq \beta_{\a,\bu_2}$.}
We also note that given any vector of reserve prices $\b r$, the
revenue that can be obtained from $\b r$ only depends on the
reserve price of  the highest and second highest bidders that clear the
reserve prices.

Next, we   formally define the notion of \emph{valid profile} and show that the
Problem \eqref{eq:opt}  can be relaxed to {find the best consistent distribution  over 
valid profiles in each auction.} {To define valid profiles, we assume that in each auction $\a$, we have two auxiliary buyers $\bu_0$ and $\bu_{00}$ who always bid zero. That is,  $\bu_{00},\bu_0 \in \B$, and $\beta_{\a,\bu_{0}} = \beta_{\a,\bu_{00}}=0$ for any $\a\in \A$.}

\begin{definition}[Valid Profiles]\label{def:profile}
  {We define the set of valid profiles for auction $\a$ as the set
  $\mathcal P_a$ consisting of all tuples $(\bu_1, \bu_2, r_1, r_2) \in \B
  \times \B \times \res \times \res$}
   satisfies the following conditions:
\begin{enumerate}
\item Bid of buyer $\bu_1$ is greater than or equal to that of buyer $\bu_2$;
  that is, $\beta_{\a, \bu_1} \ge \beta_{\a,\bu_2}$. 
\item Buyer $\bu_1$ clears his reserve; that is, $\beta_{\a,\bu_1} \geq r_1$.
\item Buyer $\bu_2$ clears his reserve; that is, $\beta_{\a,\bu_2} \geq r_2$.
\end{enumerate}
For any given profile  $p \in \mathcal P_a$, we define  $\rev{\a}{p}
  := \max(\beta_{\a, \bu_2}, r_1)$.
\end{definition}

{We note that any valid profile  corresponds to at least one vector of reserve prices.  To see why,  observe  that we can always  obtain $p=(\bu_1, \bu_2, r_1, r_2)$ by setting $r_{\bu_1} = r_1$, $r_{\bu_2} = r_2$, and $r_{\bu} = \infty$ for any  $\bu\ne \bu_1, \bu_2$. Of course, there may exist other vectors of reserve prices that lead to the same profile. We note that by adding buyers $\bu_0$ and $\bu_{00}$ to $\B$, we can  define valid profiles to  represent the cases   in which  less than two buyers cleared their reserve prices. We present the cases with    one (respectively zero)
cleared buyer with  valid profile of $(\bu_1, \bu_0, r_1, 0)$ (respectively $(\bu_0, \bu_{00}, 0,0)$). }

\begin{definition}[Profiles Associated with Reserve Prices]\label{def:profile-association} Given a vector of reserve prices $\b r$ we say a valid profile $p = (\bu_1, \bu_2, r_1, r_2)$ is the unique profile associated with $\b r$ in an auction $\a \in \A$ if and only if the following condition hold. After applying the reserve prices $\b r$, buyer $\bu_1$ with reserve $r_1$
  and buyer $\bu_2$ with reserve $r_2$ have the highest and second highest
  cleared bids in auction $\a$, respectively.
\end{definition}

Given a vector of reserve prices $\b r$ and an auction $\a$,   let $p$ be the profile associated with $\b r$ in $\a$.  Then, with a slight abuse of notation, we define $\rev{\a}{\b
  r}= \rev{\a}{p}$.

\medskip

 We are now ready to describe our LP. The LP will have two sets of variables:
\begin{enumerate}
\item 
For any auction $\a\in \A$ and any valid profile $p\in \mathcal P_{\a}$, define
    {a variable $s_{\a, p} \geq 0$ such that $\sum_{p \in \mathcal P_{\a}} s_{\a,
    p} \leq 1$.} This variable represents a probability distribution over valid profiles 
    in auction $\a$. {We refer to $\{s_{\a, p}| \a \in\A, p\in  \mathcal P_{\a} \}$ as a profile-weight.}
\item  For any buyer $\bu\in \B$ and reserve price $r\in \res$, {define a variable
  $\q_{\bu, r} \geq 0$ such that $\sum_{r \in \res} \q_{\bu, r} = 1$. This
    variable represents be the probability} that buyer $\bu$ is assigned a
    reserve price of $r$. 
\end{enumerate}

We now discuss the LP constraints. {We add constraints relating $s_{\a, p} $ and $\q_{\bu, r} $
which will ensure the consistency of probability distributions across all profiles.
To define this set of constraints,  for every $\bu \in \B$, $\a \in \A$, 
and $r \in \res$, we define  set}
 \begin{align}\label{eq:Qbr}\mathcal{Q}_{\bu,r,\a} := \{p= (\bu_1, \bu, r_1, r):
 p\in \mathcal{P}_{\a}\} \cup \{p =(\bu, \bu_2, r, r_2): p\in
 \mathcal{P}_{\a}\}\,,\end{align} 
{which corresponds to all valid profiles of auction $\a$ that assign reserve
$r$ to buyer $\bu$. A natural constraint to add is that the total probability
assigned to profiles in $\mathcal{Q}_{\bu,r,\a} $ is at most the probability
that buyer $\bu$ is assigned to reserve price  $r$. That is,
$$\sum_{p \in \mathcal{Q}_{\bu,r,\a}} s_{\a,p} \leq q_{\bu,r}\,.$$
Finally,  we can put it all together in the following LP:}
  \begin{align}
		&\max_{\b q, \b s} &&\ \sum_{\a\in \A}\sum_{p\in \mathcal{P}_{\a}} s_{\a,p} \cdot \rev{\a}{p}\nonumber	 &&\\\nonumber
		&\text{s.t.} \qquad 
		& & \sum_{p \in \mathcal{P}_{\a}} s_{\a,p} \leq 1 && \forall{\a}: \a \in \A\\\nonumber
    & && \sum_{p\in \mathcal{Q}_{b,r{,\a}}} s_{\a,p} \leq \q_{\bu,r} \qquad && \forall{\a, \bu, r}: \bu\in \B, r\in \res, \a\in \A\\ \nonumber
		& && \sum_{r\in R} \q_{\bu,r} {=} 1 && \forall{\bu}: \bu \in \B	\\
		&&& s_{\a,p} \geq 0  &&  \forall{\a, p} :  \a \in \A, \ p \in \mathcal{P}_{\a}\label{LP} \tag{\text{\sf{\small{Profile-LP}}}}
  \end{align}
 We start by
noting that the LP is a relaxation of the Problem \eqref{eq:opt}:

\begin{lemma}[Upper bound on Revenue]\label{lemma:lp}
  The solution of \ref{LP} is an upper bound to $ \esp^{\star}$, i.e., 
  the optimal value of Problem \ref{eq:opt}.
  That is, 
  \[{\esp^{\star}}\le {\ref{LP}} \,.\]
\end{lemma}

\begin{proof} Given reserve prices $\b r^{\star}$ such that $\esp^{\star} = \sum_{\a} \rev{\a}{\b
  r^\star}$, we  construct a feasible solution to the LP as
  follows.  For each $\a \in \A$, we let $s_{\a,p}=1$ for the profile $p$
  corresponding to $\b r^{\star}$ (according to Definition  \ref{def:profile-association})
  and
  $s_{\a,p} =0$ for all remaining profiles. Further, we let $q_{\bu, r^{\star}_{\bu}} =
  1$ and $q_{\bu, r} = 0$ for all remaining reserves. It is straightforward to
  verify that it is a feasible solution to the \ref{LP} and that $ \sum_{\a\in
  \A}\sum_{p\in \mathcal{P}_{\a}} s_{\a,p} \cdot \rev{\a}{p} = \esp^{\star}$.
 \end{proof}

\begin{theorem}[Integrality Gap of \ref{LP}] \label{thm:integrality_gap}
{There exists a data-set of bids $\{\beta_{\a, \bu}\}_{\a\in \A, \bu\in \B}$ for which the integrality gap of the LP is at least $0.828$. That is, 
\[ {\esp^{\star}}\le 0.828\cdot ({\ref{LP}}) \,.\]}
\end{theorem}
{The proof of Theorem \ref{thm:integrality_gap} is given in Section \ref{sec:gap}.}

\section{Profile-based LP-rounding (Pro-LPR) Algorithm}
\label{section:rounding}

In this section, we present an algorithm, called Profile-based LP-rounding
(\name{}),  that uses the optimal solution of   (\ref{LP}), $\b{\lps}$, to
devise reserve prices. The algorithm is presented below.
\medskip
\medskip

\begin{center}
\fbox{ 
\begin{minipage}{0.9\textwidth}
{
\parbox{\columnwidth}{ \vspace{1em}
\textbf{Profile-based LP-rounding
(\name{}) Algorithm:}\vspace{0.5em}

{Let 
${\b{\lps}}$ and $\b{q^{\star}}$ be the optimal solution of (\ref{LP}). Then,}

\begin{itemize}
  \item { {Rounding procedure: For}
     each buyer $\bu\in \B$,   independently  sample reserve price $r \in \res$ with probability
    proportional to $q_{\bu, r}^{\star}$.}  \item {{Let $\b{z}$ be the vector of all zero reserves.} Output the best of $\rr$
    and $\b{z}$, i.e.,  $$\rlp = \arg \max_{\b{r} \in  \{ \b{z}, \rr\}}
    \rev{}{\b{r}}\,.$$}
\end{itemize}
}}
\end{minipage}
}
\end{center} 
\medskip
\medskip
{In the \name{} algorithm, we first round the optimal solution of the  (\ref{LP}) to construct reserve prices $\rr$. To do so, for each buyer $\bu\in \B$, we independently sample reserve price $r\in \res$ with probability $q_{\bu, r}^{\star}$, where $\b{q^{\star}}$ (and ${\b{\lps}}$) is the optimal solution of  the  (\ref{LP}). We then compare revenue under $\rr$ with that under  the zero reserve prices, and return the one that obtains higher revenue. The retuned vector of reserve prices is denoted by $\rout$.}

{We now proceed to analyze our algorithm. We  show that $\E[\rev{}{\rlp}]$
is at least a $\apxf$ fraction of the solution of the \ref{LP} and hence at
most $\apxf \cdot \esp^{\star}$, where the expectation is with respect to the randomness in the \name{} algorithm. As we show in Lemma \ref{lem:initial}, one of the biggest strengths of our LP
formulation is that it allows the analysis to decouple the effect of rounding for
each individual auction.} {In this lemma, roughly speaking, we present two conditions under which the \name{} algorithm has a good performance. In these  conditions, for each auction $\a\in \A$ and $t\ge 0$, we compare  the probability that $\Rev_{\a}(\rr)$ is at least $t$, i.e., $\Pr \left[ \Rev_{\a}(\rr) \geq t \right]$, with $\sum_{\{p: p \in \mathcal P_{\a}, \Rev(p) \geq t\}} s_{\a,p}^{\star} $, which is  the sum of the optimal weight of (valid) profiles in auction $\a$ that obtains a revenue of at least $t$. Here, $\Rev_{\a}(\rr)$ is the revenue in auction $\a$ under reserve prices $\rr$. Intuitively, the smaller the gap between $\Pr \left[ \Rev_{\a}(\rr) \geq t \right]$ and the aforementioned summation, the better the \name{} algorithm performs. 
 Lemma \ref{lem:initial} makes this statement  formal 
by   considering  the high revenue 
case of $t\ge \beta^{(2)}_{\a}$ (first condition) and   the low revenue case of $t< \beta^{(2)}_{\a}$ (second condition), where $\beta^{(2)}_{\a}$ is the second highest bid in auction $\a$. Note that when the reserve price for the buyer with the highest bid in auction $\a$ is set too high, revenue of this auction can be indeed less than the second highest submitted bid $\beta^{(2)}_{\a}$.
}

\begin{lemma}[Two Conditions] \label{lem:initial}
Let 
${\b{\lps}}$ and $\b{q^{\star}}$ be the optimal solution of (\ref{LP}) and $\rr$ be
    a random reserve price obtained from the   rounding procedure. If there exists a constant $c > 0$ such that for any $t \geq 0$ and
  any auctions $\a\in \A$, we have 
  \begin{align}\sum_{\{p: p \in \mathcal P_{\a}, \Rev(p) \geq t\}} s_{\a,p}^{\star} - \Pr \left[ \Rev_{\a}(\rr) \geq t \right] ~ &\leq~
  0 \quad \text{ for } t > \beta^{(2)}_{\a} \label{eq:condition1}\\
  \sum_{\{p: p \in \mathcal P_{\a}, \Rev(p) \geq t\}} s_{\a,p}^{\star} - \Pr\left[ \Rev_{\a}(\rr) \geq t \right] ~& \leq~
  c \quad \text{ for } t \le  \beta^{(2)}_{\a}\,,\label{eq:condition2}\end{align}
  then  \name{} algorithm is a $(1+c)^{-1}$-approximation.  That is, it obtains at least  $(1+c)^{-1}$ fraction of  the optimal value  of Problem  \ref{eq:opt}. Here, $\beta^{(2)}_{\a}$ is the second highest bid in auction $\a$ and $\Rev_{\a}(\rr)$ is the revenue in auction $\a$ under reserve prices $\rr$.
 \end{lemma}
 
 \begin{proof}
 By integrating over $t$ in Equations (\ref{eq:condition1}) and (\ref{eq:condition2}) and adding them up, we get 
 \begin{align}&\int_{\beta^{(2)}_{\a}}^{\infty}\Big(\sum_{\{p: p \in \mathcal P_{\a}\Rev(p) \geq t\}} s_{\a,p}^{\star} - \Pr \left[ \Rev_{\a}(\rr) \geq t \right]\Big) dt \nonumber \\
 &+ \int_{0}^{\beta^{(2)}_{\a}}\Big(\sum_{\{p: p \in \mathcal P_{\a}\Rev(p) \geq t\}} s_{\a,p}^{\star} - \Pr \left[ \Rev_{\a}(\rr) \geq t \right]\Big) dt ~\le~ c \cdot\beta^{(2)}_{\a}\,.\nonumber 
 \end{align}
 This is simplified as follows 
 \begin{align} 
  \sum_{p\in \mathcal P_{\a}}s_{\a,p}^{\star} \Rev_{\a}(p) - \E[\Rev_{\a}(\rr)] \le  c \cdot \beta^{(2)}_{\a}\,. \label{eq:int}\end{align}
Note that by Lemma \ref{lemma:lp},  the optimal value  of Problem  \ref{eq:opt}, denoted by $ \esp^{\star}$, is upper bounded by $\rev{}{\b{\lps}}$.  That is, 
\begin{align}\esp^{\star} \leq \rev{}{\b{\lps}} = \sum_{\a\in \A} \sum_{p \in \mathcal P_{\a}}s_{\a,p}^{\star} \Rev(p)\,. \label{eq:upperbound}\end{align}
Further, the revenue of  \name{} algorithm, i.e.,  $\E[\Rev(\rout)]$, is lower bounded by
\begin{align}\E[\Rev(\rout)] \ge \max \Big(\sum_{\a\in A}\beta^{(2)}_{\a}, \E[\Rev(\rr)] \Big)\,. \label{eq:lowerbound}\end{align}
To see why this holds note that  \name{} algorithm returns the best of reserve price $\rr$ and all zero prices, where the revenue under all zero prices is the sum of the second highest highest bids $\sum_{\a\in A}\beta^{(2)}_{\a}$. \revision{By using  Equations \eqref{eq:int}, (\ref{eq:upperbound}), and (\ref{eq:lowerbound}), we have
\begin{align}\nonumber
\esp^{\star} - \E[\Rev(\rout)]& \le c\sum_{\a\in \A} \beta_{\a}^{(2)} \\
&\le c  \E[\Rev(\rout)]\,.
\end{align}
Putting these together, we have
\[\E[\Rev(\rout)] \ge
 \frac{1}{1+c}\cdot \esp^{\star}\,,\]
which is the desired result.}
\Halmos \end{proof}

{In the next lemma, we show that the first condition holds. We then dedicate the next section to identifying constant $c$ in the second condition. The proof of  the lemma is based on the observations that (i) revenue of any valid profile $(\bu_1, \bu_2, r_1, r_2)$ is greater than  $\beta^{(2)}_{\a}$ if buyer $\bu_1=\bu_{\a}^{(1)}$  and his reserve $r_1> \beta^{(2)}_{\a}$, and (ii) revenue of auction $\a$ under reserve prices  $\rr$ is greater than $\beta^{(2)}_{\a}$ if the reserve price of buyer $\bu_{\a}^{(1)}$ is less than or equal to his bid and greater than the second highest bid $\beta^{(2)}_{\a}$. Here, buyer $\bu_{\a}^{(1)}$ is the buyer with the highest bid in auction $\a$.}

\begin{lemma} [First Condition Holds]\label{lemma:first_cond}
Let $\b{\lps{}}$ denote an optimal solution of \ref{LP} and 
  $\rr$ be
    a random reserve price obtained from the   rounding procedure in the  \name{} algorithm.
For any auction $\a \in A$, we have
  \begin{align}\sum_{\{p: p \in \mathcal P_{\a}, \Rev(p) \geq t\}} s_{\a,p}^{\star} - \Pr \left[ \Rev_{\a}(\rr) \geq t \right] ~ &\leq~
  0 \quad \text{ for } t > \beta^{(2)}_{\a}. \label{eq:condition11}\end{align}
\end{lemma}

\begin{proof}
The first term in the l.h.s. of (\ref{eq:condition11}) can be written as 
\begin{align}
\sum_{\{p: p \in \mathcal P_{\a}, \Rev(p) \geq t\}} s_{\a,p}^{\star} = \sum_{ \{p : p \in \mathcal P_{\a}, p=(\bu_{\a}^{(1)}, \bu_2, r, r_2),~ r\ge t  \}} s_{\a,p}^{\star} ~\le ~ \sum_{r \ge t } \q_{\bu_{\a}^{(1)}, r} = \Pr\left[ \Rev_{\a}(\rr) \geq t \right]\,,\label{eq:verify_first_condition}
\end{align}
where the first equation holds because revenue of a profile $p\in \mathcal P_{\a}$ is $t > \beta^{(2)}_{\a} $ if and only if the bidder with the highest bid in auction $\a$, i.e., $\bu_{\a}^{(1)}$, is assigned a reserve price $t> \beta^{(2)}_{\a}$ and the bid of this bidder is greater than $t$.  The second equation holds because of the second set of constraints of (\ref{LP}). The last equation follows from the construction of reserve prices $\rr$. Note that Equation (\ref{eq:verify_first_condition}) verifies condition (\ref{eq:condition11}).
\Halmos\end{proof}

\subsection{Bounding the Constant in the Second Condition}

{We start by noting that the second condition in Lemma \ref{lem:initial} holds
trivially for $c=1$, which recovers the same approximation factor of $1/2$ of
\cite{RW16}. For the rest of the paper, we will improve past $1/2$ by
constructing a non-linear mathematical program to optimize $c$ and then applying the
first order conditions in non-linear programming to bound the optimal solution.}
In Lemma \ref{lemma:apxf}, we show that
   \[c = \max_{\qq\in [0,1]}\optf{\qq}{m}\,,\]
   where for any real number $\qq\in [0, 1]$, $\optf{\qq}{m}$ is defined as follows
    \begin{align} 
\optf{\qq}{m}~=~ \max_{\b{x}\ge 0} &~~~e^{\qq-1}\left[ \prod_{i\in [n]} (1-x_{i}) + \sum_{i\in [n]} x_{i} 
\prod_{j\in [n] , j\neq i} (1-x_{j}) \right]\nonumber  &\\\nonumber 
 \text{s.t.} \quad &\frac{1}{2} \sum_{i\in [n]} x_{i} =\qq \\
&  x_{i} \leq \qq,  \quad \forall i \in [n]\,.
\label{eq:opt_q_m}
\end{align}
{Here, $n$ is the number of buyers.} Characterizing  $\optf{\qq}{m}$ is technically involved and because of that its details  is postponed to Section~\ref{sec:approx_factor}. There, we
show that for  any number of buyers  $n\ge 2$ and any real number $\qq\in [0, 1]$,
$$\optf{\qq}{m} \le 2 \left(\sqrt{2}-1\right) e^{\sqrt{2}-2} \approx 0.4612.$$
Then, invoking Lemmas \ref{lem:initial} and \ref{lemma:first_cond}, this leads
to the approximation factor of $\frac{1}{1+0.4612} \approx 0.6844$, which is the desired result. 

In the next lemma, we formally state the relationship between $\optf{\qq}{m}$ 
and the approximation factor of our algorithm.

\begin{lemma}[Second Condition]\label{lemma:apxf}
Let $\b{\lps{}}$ denote an optimal solution of \ref{LP} and 
$\rr$  be
    a random reserve price obtained from the   rounding procedure in the  \name{} algorithm. Let 
\[c = \max_{\qq\in [0,1]}\optf{\qq}{m}\,. \]
 Then, for any auction $\a \in A$, the following equation holds.
  \begin{align} \nonumber
  \sum_{\{p: p \in \mathcal P_{\a}, \Rev(p) \geq t\}} s_{\a,p}^{\star} - \Pr\left[ \Rev_{\a}(\rr) \geq t \right] ~& \leq~
  c \quad \text{ for } t \le  \beta^{(2)}_{\a}\,. \end{align}
\end{lemma}

{The formal proof of Lemma \ref{lemma:apxf} due to being lengthy is deferred to Subsection~\ref{subsection:proof-apxf}, however we provide some intuition here. For each buyer $\bu$, we consider two disjoint subsets of valid profiles $(\bu_1,\bu_2, r_1, r_2)$ such that (i) revenue of any profile in these subsets is greater than or equal to $t$, where $t\le \beta_{\a}^{(2)}$, and (ii) either $\bu_1$ or $\bu_2$ is equal to buyer $\bu$. 
The factor $1/2$ in the constraint of Problem (\ref{eq:opt_q_m}) is the artifact of the definition of the subsets and how the  summation $ \sum_{\{p: p \in \mathcal P_{\a}, \Rev(p) \geq t\}} s_{\a,p}^{\star}$  can be written as a function of the optimal weight of the profiles in these subsets; see Equation (\ref{eq:question}) in the proof. We then express  $\Pr\left[ \Rev_{\a}(\rr) \geq t \right]$ as the probability of the  union of two events, where the first event happens  if there is at least one cleared buyer with reserve price greater than $t$, and the second event happens if there are at least two buyers with cleared bids of at least $t$. We then write this probability as a function of the profile weights in the subsets  by taking advantage of the fact that in our rounding procedure,  reserve prices are independent across buyers. In particular, we show that this probability is at least one minus the left hand side of the equation in Lemma \ref{lemma:moving_buyers}, stated below.   We  invoke Lemma \ref{lemma:moving_buyers} to complete the proof. Observe that the objective function of Problem (\ref{eq:opt_q_m}) bears significant resemblance to that of  the right hand of the equation in Lemma~\ref{lemma:moving_buyers}.   }

\begin{lemma}\label{lemma:moving_buyers} {Consider a set $\hat \B\subseteq \B$ with $|\hat \B|\ge 2$.\footnote{Note that because of the auxiliary buyers, $|\B|\ge 2$. } Given fixed $x_{1,\bu}, x_{2,\bu}$ with $\bu \in \hat \B$  and $x_{1,\bu} +
  x_{2,\bu} \leq 1$, the following inequality holds:
  $$  \prod_{\bu\in \hat \B} (1- x_{1,\bu}-  x_{2,\bu}) + \sum_{\bu\in \hat \B}  x_{2,\bu} \prod_{\bu' \neq \bu} (1-
  x_{1,\bu'}- x_{2,\bu'}) \leq  \prod_{\bu\in \hat \B} (1- x_{1,\bu}) \left[\prod_{\bu\in \hat \B} (1-  x_{2,\bu}) + \sum_{\bu\in \hat \B}  x_{2,\bu} \prod_{\bu' \neq \bu} (1- x_{2,\bu'})\right].  $$}
\end{lemma}

\begin{proof}Given a partition of {$\hat \B$} in two sets $\B_1, \B_2$, define the
  following function:
  $$\begin{aligned}\Phi(\B_1, \B_2) = & \prod_{\bu \in \B_1} (1- x_{1,\bu}) (1- x_{2,\bu}) \prod_{\bu \in
    \B_2}  (1- x_{1,\bu} - x_{2,\bu})  + \\ &  \sum_{\bu \in \B_1 \cup \B_2} x_{2,\bu} \left[ \prod_{\bu' \in
  \B_1, \bu' \neq \bu} (1- x_{1,\bu}) (1- x_{2,\bu})  \prod_{\bu' \in \B_2, \bu' \neq \bu} (1- x_{1,\bu} -
  x_{2,\bu}) \right]. \end{aligned}  $$
  The main claim in the lemma is that $\Phi(\B, \emptyset) \geq \Phi(\emptyset,
  \B)$. We will show that for any $\B_1, \B_2$ and $\hat \bu \in \B_2$, we have
  $$\Phi(\B_1, \B_2)  \leq  \Phi(\B_1 \cup \{\hat \bu\}, \B_2 \setminus \{\hat \bu\}) $$
  and the claim will follow by moving the elements from $\B_2$ to $\B_1$ one by
  one. To simplify notation, define
  $$w =  \prod_{\bu \in \B_1} (1- x_{1,\bu}) (1- x_{2,\bu}) \prod_{\bu \in \B_2
  \setminus \{ \hat \bu \}}  (1- x_{1,\bu} - x_{2,\bu})\,. $$
  Now we can write:
  $$\Phi(\B_1, \B_2) = w \cdot (1-x_{1,\hat \bu}-x_{2,\hat \bu}) + w \cdot x_{2,\hat \bu}
  + \sum_{\mathclap{\bu \in \B_2; \bu \neq \hat \bu}} \;\;\, w \cdot \frac{1-x_{1,\hat \bu}-x_{2,\hat
  \bu}}{1-x_{1,\bu}-x_{2,\bu}} \cdot x_{2,\bu} + \sum_{\bu \in \B_1} w \cdot \frac{1-x_{1,\hat \bu}-x_{2,\hat
  \bu}}{(1-x_{1,\bu})(1-x_{2,\bu})} \cdot x_{2,\bu}$$
  and
  \begin{align*}\Phi(\B_1 \cup \{\hat \bu\}, \B_2 \setminus \{\hat \bu\}) &= w \cdot
  (1-x_{1,\hat \bu})(1-x_{2,\hat \bu}) + w \cdot x_{2,\hat \bu}
  \\
  &+ \sum_{\bu \in \B_2; \bu \neq \hat \bu} w \cdot \frac{(1-x_{1,\hat \bu})(1-x_{2,\hat
  \bu})}{1-x_{1,\bu}-x_{2,\bu}} \cdot x_{2,\bu} + \sum_{\bu \in \B_1} w \cdot
  \frac{(1-x_{1,\hat \bu})(1-x_{2,\hat
  \bu})}{(1-x_{1,\bu})(1-x_{2,\bu})} \cdot x_{2,\bu}\end{align*}
Our goal here is to show $\Phi(\B_1, \B_2)  \leq  \Phi(\B_1 \cup \{\hat \bu\}, \B_2
  \setminus \{\hat \bu\}) $. 
 We start with  comparing the first two terms of $\Phi(\B_1, \B_2) $ and   $\Phi(\B_1 \cup \{\hat \bu\}, \B_2
  \setminus \{\hat \bu\})$:
   $$ w \cdot (1-x_{1,\hat \bu}-x_{2,\hat \bu}) + w \cdot x_{2,\hat \bu}
 = w \cdot (1-x_{1,\hat \bu}) \leq w \cdot (1-x_{1,\hat \bu} + x_{1,\hat \bu}
  x_{2,\hat \bu}) =  w \cdot
  (1-x_{1,\hat \bu})(1-x_{2,\hat \bu}) + w \cdot x_{2,\hat \bu}\,.$$
  We can compare the remaining terms one by one using the fact that:
  $$ 1-x_{1,\hat b}-x_{2,\hat \bu} \leq (1-x_{1,\hat \bu})(1-x_{2,\hat \bu})\,. $$
  This concludes that $\Phi(\B_1, \B_2)  \leq  \Phi(\B_1 \cup \{\hat \bu\}, \B_2
  \setminus \{\hat \bu\}) $ as desired.
   \end{proof}

\subsection{Proof of Lemma~\ref{lemma:apxf}}\label{subsection:proof-apxf}
We start with a few definitions. Consider a certain auction $\a \in \A$ and all of its valid profiles $p \in \mathcal P_{\a}$. Fix some threshold $t\le \beta^{(2)}_{\a}$
and an optimal  solution  of  (\ref{LP}), denoted by $\b{\lps{}}$.  {Let set $\B_{\a,t}$ be the set of buyers whose  bid in auction $\a$ is at least $t$:
$$\B_{\a,t} := \{\bu \in \B; \beta_{\a, \bu} \geq t\}\,.$$ 
Note that this set is not empty because $t\le \beta^{(2)}_{\a}$. In fact, $|\B_{\a,t}|\ge 2$, as buyers with the highest and second-highest bids belong to this set. (Recall that because of the auxiliary buyers, $|\B|\ge 2$.)
A crucial observation is that the reserve assigned to any buyer $\bu \notin \B_{\a,t}$ does not affect the event $\Rev_{\a}(\rr) \geq t$ since such buyers can be neither the winner nor the price setter in an auction with revenue of at least $t$. }
Consider a buyer ${\bu \in \B_{\a,t}}$. Then, define  
  \[\mathcal X'_{1, \bu} =\{ p = (\bu, \bu_2, r_1, r_2): p\in \mathcal P_{\a},~  r_1\ge t\}\] 
  \[\mathcal X''_{1, \bu} =\{ p = (\bu_1, \bu, r_1, r_2): p\in \mathcal P_{\a},~
  r_1< t \text{ and } r_2 \geq t\}\] 
  \[\mathcal X_{2, \bu} =\{ p = (\bu_1, \bu_2, r_1, r_2): p\in \mathcal P_{\a},
  \bu \in \{\bu_1, \bu_2\}, r_1, r_2 < t \text{ and }  \beta_{\a,\bu_2} \ge t\} \] 
  and then set
  \[x_{1, \bu} = \sum_{p\in \mathcal X'_{1, \bu} \cup X''_{1,\bu} } s_{\a,p}^{\star}  \quad \text{and} \quad x_{2, \bu} = \sum_{p\in \mathcal X_{2, \bu} } s_{\a,p}^{\star} \,.\]
  {We note that $\mathcal X'_{1, \bu}$ is the set of all valid profiles $p = (\bu, \bu_2, r_1, r_2)$ in which  reserve of buyer $\bu$ is at least $t$.  $\mathcal X''_{1, \bu}$ is the set of all valid profiles $p = (\bu_1, \bu, r_1, r_2)$ in which reserve of  buyer $\bu_1$ is less than $t$ and reserve of buyer $\bu$ is greater than or equal to $t$. Observe that for all the profiles $p$ in $\mathcal X'_{1, \bu}\cup \mathcal X''_{1, \bu}$, reserve of buyer $\bu$ is at least $t$. This implies that for all of these profiles, $\Rev(p)\ge t$. We also note that $\mathcal X_{2, \bu}$ is the set of all valid profiles $p = (\bu_1, \bu_2, r_1, r_2)$ such that buyer $\bu\in \{\bu_1, \bu_2\}$ and bid of buyer $\bu_2$ is at least $t$. Again, it is easy to see that for any valid profile $p\in \mathcal X_{2, \bu}$, we have $\Rev(p)\ge t$. Finally, we point that while any profile $p$ in $\mathcal X_{2, \bu}$ and $\mathcal X'_{1, \bu}\cup \mathcal X''_{1, \bu}$ has $\Rev(p) \ge t$, by construction, $\mathcal X_{2, \bu}$ and $\mathcal X'_{1, \bu}\cup \mathcal X''_{1, \bu}$ are disjoint. Therefore, we have}
  \begin{align} & \sum_{\{p: ~p \in \mathcal P_{\a}, ~\Rev(p) \geq t\}} s_{\a,p}^{\star} ~=~ \sum_{{\bu \in \B_{\a, t}}} x_{1,\bu} + \frac{1}{2} 
  \sum_{{\bu \in \B_{\a, t}}} x_{2,\bu} \,,\label{eq:question} \end{align}
  where the coefficient $\frac{1}{2}$ accounts for double-counting.  
  Define $y_{1, \bu}$ as the probability that  
  the sampled reserve of buyer $\bu$, i.e., $\rrb$,  is in $[t,\beta_{\a,\bu}]$ and $y_{2,\bu}$ as the
  probability that the sampled reserve $\rrb$ is in $[0,t)$.
  By the sampling procedure we know that:
  \begin{align}y_{1, \bu} \geq x_{1,\bu} \quad \text{and}\quad y_{2,\bu} \geq x_{2,\bu}\,. \label{eq:ineq}\end{align}
  Observe that $\Rev_{\a}(\rr) \geq t$ iff at least one of the two following events happen.
\begin{enumerate}
\item []\textbf{Event $\mathcal E_1$:} There exists a buyer {$\bu \in \B_{\a, t}$} with a reserve of at least $t$ whose bid is cleared. 
\item [] \textbf{Event $\mathcal E_2$:} There are at least two buyers {$\bu_1, \bu_2 \in \B_{\a, t}$} with cleared bids of at least $t$.
 \end{enumerate}
  Precisely, 
  \begin{equation} \label{eq:pos} \Pr\left[ \Rev_{\a}(\rr) \geq t \right] ~=~
 \Pr[\mathcal E_1 \text{ or }  \mathcal E_2] ~=~ \Pr[\mathcal E_1] + \Pr[ \mathcal
 E_2 \text{ and } \bar{\mathcal E}_1] ~=~ \Pr[\mathcal E_1] +  \Pr[ \bar{\mathcal E}_1] \Pr[ \mathcal
 E_2 \vert \bar{\mathcal E}_1]\,,\end{equation}
where 
$$ \Pr[\mathcal E_1] ~=~ 1 - \prod_{{\bu \in \B_{\a, t}}} (1- y_{1,\bu}),  $$
and 
  $$ \Pr[ \mathcal E_2 \vert \bar{\mathcal E}_1] ~=~ 1 - \prod_{{\bu \in \B_{\a,t}}} (1- \tilde y_{2, \bu})
  - \sum_{{\bu \in \B_{\a,t}}} \tilde y_{2, \bu} \prod_{\bu' \neq \bu} (1- \tilde y_{2, \bu'}) \quad \text{ with }
  \quad \tilde y_{2, \bu} ~=~ \frac{y_{2, \bu}}{1-y_{1, \bu}}\,.$$
  This gives us
   $$\Pr[ \mathcal E_2 \text{ and } \bar{\mathcal E}_1] ~=~ \Pr[ \bar{\mathcal E}_1] \Pr[ \mathcal
 E_2 \vert \bar{\mathcal E}_1] ~=~  \Pr[\bar{\mathcal E}_1] -
  \prod_{{\bu \in \B_{\a,t}}} (1-y_{1,\bu} - y_{2,\bu}) - \sum_{{\bu \in \B_{\a,t}}}  y_{2, \bu} \prod_{\bu' \neq \bu} (1- y_{1, \bu'} -y_{2, \bu'})\,. $$
  Thus, by Equation (\ref{eq:pos}), we get
    $$\Pr[ \mathcal E_2 \text{ or } {\mathcal E}_1] ~=~ 1 -
  \prod_{{\bu \in \B_{\a,t}}} (1- y_{1,\bu}-  y_{2,\bu}) - \sum_{{\bu \in \B_{\a,t}}}  y_{2,\bu} \prod_{\bu' \neq \bu} (1- y_{1,\bu'}- y_{2,\bu'})\,. $$
  Now observe that the expression above, i.e., $\Pr[ \mathcal E_2 \text{ or } {\mathcal E}_1]$,  is increasing  in both $y_{1,\bu}$ and
  $y_{2,\bu}$, ${\bu \in \B_{\a,t}}$. {To see why $\Pr[ \mathcal E_2 \text{ or } {\mathcal E}_1] $ is increasing in $y_{2,\bu}$, note that 
    \begin{align}\nonumber \frac{ \partial (\Pr[ \mathcal E_2 \text{ or } {\mathcal E}_1])} {\partial   y_{2,\bu}}
   &~=~ 
   \sum_{\bu'\in {\B_{\a,t}}, \bu' \neq  \bu}  y_{2,\bu'} \prod_{\bu'' \neq \bu,  \bu'} (1- y_{1,\bu''}- y_{2,\bu''}) \ge 0\,.
     \end{align}}
  This and Equation \eqref{eq:ineq} imply  that:
    $$\Pr[ \mathcal E_2 \text{ or } {\mathcal E}_1] ~\geq ~ 1 -
  \prod_{{\bu \in \B_{\a,t}}} (1- x_{1,\bu}-  x_{2,\bu}) - \sum_{{\bu \in \B_{\a,t}}}  x_{2,\bu} \prod_{\bu' \neq \bu} (1-
  x_{1,\bu'}- x_{2,\bu'})\,. $$
  We now invoke Lemma \ref{lemma:moving_buyers}, stated earlier, to get 
\begin{align}\Pr[ \mathcal E_2 \text{ or } {\mathcal E}_1] ~\geq ~ 1 -
  \prod_{{\bu \in \B_{\a,t}}} (1- x_{1,\bu}) \left[\prod_{{\bu \in \B_{\a,t}}} (1-  x_{2,\bu}) + \sum_{{\bu \in \B_{\a,t}}}  x_{2,\bu} \prod_{\bu' \neq \bu} (1- x_{2,\bu'})\right]. \label{eq:union}\end{align}
Using Equations (\ref{eq:question}), (\ref{eq:union}), and  (\ref{eq:pos}), we have 
\begin{align}\nonumber & \;\;\;\;\;\;\; \;\sum_{\mathclap{\{p: p \in \mathcal P_{\a}, \Rev(p) \geq t\}}} \;\;\;\;\; s_{\a,p}^{\star} -
    \Pr\left[ \Rev_{\a}(\rr) \geq t \right] ~\leq~ \\ & \quad \sum_{{\bu \in \B_{\a,t}}} x_{1,\bu} + \frac{1}{2} 
  \sum_{{\bu \in \B_{\a,t}}} x_{2,\bu} - \left(1 -
  \prod_{{\bu \in \B_{\a,t}}} (1- x_{1,\bu}) \left[\prod_{{\bu \in \B_{\a,t}}} (1-  x_{2,\bu}) + \sum_{{\bu \in \B_{\a,t}}}  x_{2,\bu} \prod_{\bu' \neq \bu} (1- x_{2,\bu'})\right] \right).\label{eq:bound} \end{align}
  We claim that for any ${\bu \in \B_{\a,t}}$, the above expression is non-decreasing in $x_{1, \bu}$. \revision{To get this, we need to show that the derivative of the above expression w.r.t.  $x_{1, \hat \bu}$ is non-negative. In the other words, we need to show the following equation holds:
  \begin{align}&1-\prod_{{\bu \in \B_{\a,t}}, \bu\ne \hat \bu} (1- x_{1,\bu}) \left[\prod_{{\bu \in \B_{\a,t}}} (1-  x_{2,\bu}) + \sum_{{\bu \in \B_{\a,t}}}  x_{2,\bu} \prod_{\bu' \neq \bu} (1- x_{2,\bu'})\right]\geq 0. \label{eq:ijefjier}
  \end{align} 
  Since $0\leq(1-x_{1,\bu})\leq 1$  for any $\bu$, it only remains to show that the value of the term in the brackets, i.e., $\prod_{{\bu \in \B_{\a,t}}} (1-  x_{2,\bu}) + \sum_{{\bu \in \B_{\a,t}}}  x_{2,\bu} \prod_{\bu' \neq \bu} (1- x_{2,\bu'})$, is always in the range of $[0,1]$. To get this, it suffices to show that there exists an event whose probability can be written as $\prod_{{\bu \in \B_{\a,t}}} (1-  x_{2,\bu}) + \sum_{{\bu \in \B_{\a,t}}}  x_{2,\bu} \prod_{\bu' \neq \bu} (1- x_{2,\bu'})$. For any $\bu$ define a Bernoulli random variable with mean $x_{2,\bu}$. Observe that the aforementioned term is equal to the probability of the event in which at most one of these variables is equal to one, assuming that they are independent.} Thus, we obtain Equation~\eqref{eq:ijefjier}, which
 allows us to assume without loss of generality
  that $\sum_{{\bu \in \B_{\a,t}}} x_{1,\bu} + \frac{1}{2} 
  \sum_{{\bu \in \B_{\a,t}}} x_{2,\bu} {=1}$. As a result, we have
   $$\begin{aligned} &  \;\;\;\;\;\;\; \;\sum_{\mathclap{\{p: p \in \mathcal P_{\a}, \Rev(p) \geq t\}}} \;\;\;\;\; s_{\a,p}^{\star} -
    \Pr\left[ \Rev_{\a}(\rr) \geq t \right] \leq   \prod_{{\bu \in \B_{\a,t}}} (1- x_{1,\bu}) \left[\prod_{{\bu \in \B_{\a,t}}} (1-  x_{2,\bu}) + \sum_{{\bu \in \B_{\a,t}}}  x_{2,\bu} \prod_{\bu'
  \neq \bu} (1- x_{2,\bu'})\right],  \end{aligned}$$
  where $\sum_{{\bu \in \B_{\a,t}}} x_{2,\bu} = 2\qq$, $\sum_{{\bu \in \B_{\a,t}}} x_{1,\bu} = 1-\qq$. Here, $\qq  \in [0,1]$.  To complete the proof,  we simply use that:
  $\prod_{{\bu \in \B_{\a,t}}} (1- x_{1,\bu}) \leq e^{-{\sum_{b \in \B_{\a,t}}} x_{1,\bu}} = e^{\qq-1}$. Given how we constructed
  the variables $x_{2,\bu}$,  we also need $x_{2,\bu} \leq \qq$. Hence, 
  $$\begin{aligned} & \;\;\;\;\;\;\; \;\; \sum_{\mathclap{\{p: p \in \mathcal P_{\a}, \Rev(p) \geq t\}}} \;\;\;\;\; s_{\a,p}^{\star} -
    \Pr\left[ \Rev_{\a}(\rr) \geq t \right] ~\leq~   e^{\qq-1} \left[\prod_{{\bu \in \B_{\a,t}}} (1-  x_{2,\bu}) + \sum_{{\bu \in \B_{\a,t}}}  x_{2,\bu} \prod_{\bu'
  \neq \bu} (1- x_{2,\bu'})\right],  \end{aligned}$$
where $\sum_{{\bu \in \B_{\a,t}}} x_{2,\bu} = 2\qq$ and  $x_{2,\bu} \leq \qq$ for any ${\bu \in \B_{\a,t}}$. 

\subsection{Approximation Factor} \label{sec:approx_factor}

In this section,  we will show that for any given $\qq\in [0, 1]$, we have
 \[\optf{\qq}{m} \leq 2 \left(\sqrt{2}-1\right) e^{\sqrt{2}-2}\,,\]
 where $\optf{\qq}{m}$ is defined in Equation (\ref{eq:opt_q_m}). {
Since the constraints of Program (\ref{eq:opt_q_m}) are linear in $x_i$'s, the
first order conditions of Karush-Kuhn-Tucker (KKT) are a necessary condition for
optimality \cite{bertsekas1999nonlinear}. Let
$$F(x, \qq) = e^{\qq-1}\left[ \prod_{i\in [n]} (1-x_{i}) + \sum_{i\in  [n]}
x_{i} \prod_{j\in [n] , j\neq i} (1-x_{j}) \right].$$
Observe that  $F(x, \qq)$ is the objective function of  $\optf{\qq}{m}$.
 Then, according to the KKT conditions, the optimal solution must satisfy the following
constraints for some $\lambda \in \mathbb{R}$, $\mu,\eta \in \mathbb{R}^n_+$:
\begin{align}\label{kkt1}
  \nabla_x F(x, \qq) + \frac{\lambda}{2} {\bf 1} - \mu + \eta = 0&\\
\label{kkt3}
\sum_{i\in [n]} x_i = \frac{1}{2}\theta&\\
\label{kkt4}
\mu_i (x_i - \qq) = 0, \quad \forall i\in[n]&\\
\label{kkt5}
\eta_i x_i = 0, \quad \forall i\in [n]&\\
\label{kkt6}
0 \leq x_i \leq  \qq, \quad \forall i\in [n]\,.&
\end{align}
}
{where ${\bf 1} \in \mathbb{R}^n$ is the vector of all one.}

{It is enough to show that $ F(x, \qq) \leq 2 \left(\sqrt{2}-1\right)
e^{\sqrt{2}-2}$ for any tuple $(x,\qq, \lambda, \mu, \eta)$ satisfying the KKT
conditions. A simple consequence of the KKT condition is the following:
\begin{lemma}[KKT Condition]\label{lemma:eq_value} If $(x,\qq, \lambda, \mu, \eta)$  satisfies the KKT conditions for
  Problem (\ref{eq:opt_q_m}), then if $x_k$ and $x_t$ are such that $0 < x_k <
  \theta$ and $0 < x_t < \theta$, then $x_k = x_t$.
\end{lemma}
}
\begin{proof} By conditions (\ref{kkt4}) and (\ref{kkt5}), we must have $\mu_k =
  \eta_k = 0$. Plugging that into condition  (\ref{kkt1}),  we
  get that \[\partial F / \partial x_k  + \lambda / 2 = 0\,.\]
  This implies that
\begin{align} \sum_{i\ne k} x_{i}
\prod_{j\neq i,k} (1-x_{j}) +\frac{\lambda}{2} =0\,. \nonumber\end{align}
Let $Q = \prod_{i\in[n]} (1-x_{i})$ and $S= \sum_{i\in[n]} \frac{x_{i}}{1-x_{i}}$. Then, the above condition can be written as
\begin{align}
 \frac{Q}{1-x_{k}} \sum_{i\ne k} \frac{x_{i}}{1-x_{i}}
 +\frac{\lambda}{2} =0
~~ \Rightarrow~~ \frac{Q}{1-x_{k}}(S-\frac{x_{k}}{1-x_{k}} )
 +\frac{\lambda}{2} =0\,.\nonumber
 \end{align}
 This is further simplified as follows
\begin{align}
  \big(SQ+\frac{\lambda}{2}\big) -\big(SQ+Q+\lambda\big) x_{k}  +\frac{\lambda}{2} x^2_{k}= 0\,.\nonumber
\end{align}
  {The polynomial $p(y) := \big(SQ+\frac{\lambda}{2}\big) -\big(SQ+Q+\lambda\big)
  y  +\frac{\lambda}{2} y^2$ is quadratic with $\frac{d^2p}{d^2y} \geq 0$ and $p(1) = -Q < 0$. {Thus, $p(y) =0$ has an unique solution with $y< 1$.}  This  implies $x_k$  is
  uniquely determined as a function of $S$, $Q$, and $\lambda$. By the
  same argument, $x_t$ is also a solution to the same equation and hence $x_k =
  x_t$.}
\Halmos\end{proof}

{Lemma \ref{lemma:eq_value} leads to the following corollary.}

{
\begin{corollary}\label{cor:opt}
  We can bound $\optf{\qq}{} \leq \max_{k \in \mathbb Z, k \geq 2}
  \max [\optfone{\qq}{k}, \optftwo{\qq}{k} ]$,  where
  $$\optfone{\qq}{k} = e^{\theta-1}  \left(
  1-\frac{2\theta}{k} \right)^{k-1} \left( 1-\frac{2\theta}{k} + 2\theta
  \right)$$
  $$\optftwo{\qq}{k}  = e^{\theta-1}  \left[ \left( 1-\frac{\theta}{k}
  \right)^{k} + 
  \theta(1-\theta)\left( 1-\frac{\theta}{k} \right)^{k-1} \right]\,.$$
\end{corollary}}

\begin{proof}
  As stated earlier, in order to maximize the objective function $\optf{\qq}{} $,  it is enough to consider feasible
  solutions
  $x$ satisfying the KKT conditions. To do so, we use Lemma \ref{lemma:eq_value} to narrow
  down such solutions.

 {
  Since for any $i\in [n]$,  $x_i \leq \theta$ and $\sum_{i\in [n]} x_i = 2\theta$, we an only have the following three cases:
  \begin{itemize}
  \item Case 1:  Two variables
  in the support have value $\theta$ and by constraint $\sum_{i\in [n]} x_i = 2\theta$, the rest of them are zero. In that case,    
  $\optf{\qq}{} =\optfone{\qq}{2}$. 
  \item Case 2: One variable has value $\theta$ and by Lemma \ref{lemma:eq_value}, 
    {the rest $n-1\ge 2$}
  variables in the support have value $\theta/(n-1)$. In that case,  
$\optf{\qq}{} =\optftwo{\qq}{n-1}$.
    \item Case 3: All variables in the support are strictly below $\theta$. In this case,  by Lemma \ref{lemma:eq_value},  $x_i =
  \theta/n$ for $n \geq 3$, and the solution is $\optf{\qq}{} = \optfone{\qq}{n}$. 
  \end{itemize}}\Halmos\end{proof}

{
\begin{lemma} 
\label{lemma:boundcase1} For any $\qq\in [0,1]$ and $k\geq 2$, we have
  $\optfone{\qq}{k} \le 2 \left(\sqrt{2}-1\right) e^{\sqrt{2}-2}$.
\end{lemma}
\begin{proof}
  For each $k \geq 0$, define $\qq^*(k)= \arg \max_{\qq\in
  [0,1]}~\optfone{\qq}{k}$. By solving $\partial  \optfone{\qq}{k} / \partial
  \theta = 0$ we obtain the following expression for $\qq^*(k)$:
\begin{align}
  k^2 (2 \qq^*(k)-1)+4 (k-1) (\qq^*(k))^2  = 0\,. \nonumber	
\end{align} 
The aforementioned equation has two solutions, only one of which is in $[0,1]$. Thus,
\begin{align} 
  \label{q^*} \qq^*(k)=\frac{k \left(k-\sqrt{k^2+4 k-4}\right)}{4-4 k}.
\end{align}
  We need to show that for any $k\geq 2$, we have $\optfone{\qq^*(k)}{k} \leq 2
  \left(\sqrt{2}-1\right) e^{\sqrt{2}-2} \approx 0.461$. For $k =2$, we have $\optfone{\qq^*(k)}{k} = 2
  \left(\sqrt{2}-1\right) e^{\sqrt{2}-2}$. 
   For $k < 40$, we
  can verify this inequality numerically. For $k \ge 40$, we define and upper
  bound:
$$U(\qq, k) = \frac{2\qq+1}{e^{\qq+1}(1-\frac{2\qq}{k})}.$$
and show that for any $\theta \in [0,1]$ and  $k \geq 40$, 
  $$\optfone{\qq}{k} \leq U(\qq, k) \leq U(\qq,40) \leq 0.459 < 2 \left(\sqrt{2}-1\right) e^{\sqrt{2}-2}\,. $$
For the first inequality note that:
\begin{align} \optfone{\qq}{k}  &= 
e^{\qq-1}\left[\left(1-\frac{2\qq}{k}\right)^{k-1}\left(1+(k-1)\frac{2\qq}{k}\right)\right]
\\
&< e^{\qq-1}\left[\left(1-\frac{2\qq}{k}\right)^{k}\left(1-\frac{2\qq}{k}\right)^{-1}
\left(1+2\qq\right)\right] \leq U(\qq, k).\end{align}
For the second inequality, we use the fact that for any $\theta$, $U(\theta, k)$ is decreasing in $k$.
 To find an upper-bound for value of $U(\qq, 40) = \frac{(2 \qq+1) }{e^{\qq+1}
  (1-\frac{\qq}{20})}$, we take derivative of that which is $$\frac{\partial
  U(\qq, 40)}{\partial \qq} = \frac{20  \left(2 \qq^2-39
  \qq+21\right)}{e^{\qq+1} (\qq-20)^2}\,.$$ By solving $\frac{\partial  U(\qq,
  40)}{\partial \qq} = 0$, we obtain that maximum of $U(\qq, 40)$ is at $\qq=
  \frac{1}{4} \left(39-\sqrt{1353}\right)$ and $$U\left(\frac{1}{4}
  \left(39-\sqrt{1353}\right), 40\right) < 0.459\,.$$ 
  This completes the proof.
\Halmos\end{proof}
}

{
\begin{lemma} \label{lemma:case2} For any $\qq\in [0,1]$ and $k \ge 2$, we have 
$\optftwo{\qq}{k}~\leq~ 0.46 < 2 \left(\sqrt{2}-1\right) e^{\sqrt{2}-2}$.
\end{lemma}
\begin{proof}
Observe that
\begin{align}e^{1-\qq}\optftwo{\qq}{k} &=
  \left(1-\frac{\qq}{k}\right)^{k}+\qq(1-\qq)\left(1-\frac{\qq}{k}\right)^{k-1} 
 \\& \leq \left(1-\frac{\qq}{k}\right)^{k}+\frac{1}{4}\left(1-\frac{\qq}{k}\right)^{k}
 = \frac{5}{4}\left(1-\frac{\qq}{k}\right)^{k}\,,
 \end{align}
 where the first inequality holds because $\max_{\theta \in [0,1]} \theta (1-\theta) =\frac{1}{4}$ and $1-\frac{\theta}{k} \le 1$.
  Finally, note that $e^{\theta-1} \cdot \frac{5}{4}(1-\frac{\qq}{k})^{k}$ is decreasing for
  $\theta \in [0,1]$, Thus, we  can bound $\optftwo{\qq}{k}$ by the value of
  $e^{\theta-1} \cdot  \frac{5}{4}(1-\frac{\qq}{k})^{k}$  at $\theta=0$
  which is $5/(4e) < 0.46$.
\Halmos\end{proof}
}

\section{Tightness of the Analysis}\label{sec:tight}

In this section, we show that the analysis of our algorithm is tight, i.e.,
we construct an example for which 
 the performance of the algorithm matches the $\apxf$ 
approximation factor.

To make the construction cleaner, we can define the weighted version of
our problem in which
each auction $\a \in \A$ has an associated weight $w_\a > 0$, and the objective is to
maximize $\sum_{\a \in \A} w_\a \cdot \rev{\a}{\b r}$. Note that if the weights
are integers, this is exactly the same as the original problem, replacing each weighted
auction by $w_\a$ (unweighted) copies. Even if $w_\a$'s are not integers, it is
easy to see that the algorithm and the analysis generalize with essentially no
change to the weighted case (the only modification involves adding weighs to the
objective function in the LP). {In other words, if the  objective were
the weighted revenue, we would still get \apxf~approximation factor by applying
a similar algorithm.} Furthermore, any lower bound to the weighted case
translates to the unweighted case by replacing a weighted auction $\a$ by
$\lfloor N w_a \rfloor$ unweighted copies for some large $N$.

\begin{theorem}[Tightness of the Analysis]\label{thm:tight}
  There is a weighted instance $\{w_\a\}_{\a \in \A}, \{\beta_{\a,
  \bu}\}_{\a \in \A, \bu \in \B}$ and an optimal LP solution $\b s, \b q$ such that
  $$\max\left( \E\left[\sum_\a w_\a \Rev_\a(\rr)\right], \sum_\a w_\a \Rev_\a(\b 0) \right) \leq \apxf \cdot \Rev(\b s).$$ 
\end{theorem}

\begin{proof}
Fix $\theta = \sqrt{2}-1$ and $c = (1-\theta^2) e^{\theta-1}$.
Consider an instance with three weighted auctions and $n = k+3$ buyers described
  by the following table:\medskip

\begin{center}
\begin{tabular}{|l||*{6}{c|}}\hline
\backslashbox{Weights $w_\a$}{Bids}
&\makebox[3em]{$\beta_{\a,1}$}&\makebox[3em]{$\hdots$}&\makebox[3em]{$\beta_{\a,k}$}
&\makebox[3em]{$\beta_{\a,k+1}$}&\makebox[3em]{$\beta_{\a,k+2}$} &  \makebox[3em]{$\beta_{\a,k+3}$}\\\hline\hline
$1/(c+1)$ & $1$ & $\hdots$ & $1$ & $1$ & $1$ & $0$\\\hline
 $c/(c+1)$ & $0$ & $\hdots$ & $0$ & $0$ & $0$ & $1+\epsilon$\\\hline
$\epsilon$  & $1+\epsilon$ & $\hdots$ & $1+\epsilon$ & $1+\epsilon$ & $1+\epsilon$ & $0$\\\hline
\end{tabular}
\end{center}

\medskip
Now, consider the following solution to the \ref{LP}. For the first
auction, 

\begin{itemize}
\item the profile $p = (i,\bu_{0},1,0)$ has $s_{\a,p} = (1-\theta)/k$ for $i \in [k]$. In
this profile, the $i$-th buyer is reserve priced at $1$ and the second buyer is
the dummy buyer;
\item the profile $p = (k+1, k+2, 0, 0)$ has weight $s_{\a,p} = \theta$. In this
profile, both buyers $k+1$ and $k+2$ have zero reserve prices. Observe that  the revenue under this profile  is
$1$ due to the highest second price.
\end{itemize}

For the second auction, we consider only one profile:
\begin{itemize}
\item the profile {$p = (k+3,\bu_{0},1+\epsilon,0)$} has $s_{\a,p} = 1$. In
this profile, the $(k+3)$-th buyer is reserve prices at $1$ and the second buyer is
the dummy buyer. 
\end{itemize}

And for the third auction, we have:

\begin{itemize}
\item the profile $p = (i,\bu_0,1+\epsilon,0)$ has $s_{\a,p} = \theta/k$ for $i \in [k]$. In
this profile, the $i$-th buyer is reserve priced at $1+\epsilon$ and the second buyer is
the dummy buyer.
\item the profile $p = (k+1, k+2, 1+\epsilon, 1+\epsilon)$ has weight $s_{\a,p} = 1- \theta$. In this
profile, both buyers $k+1$ and $k+2$ have reserve price $1+\epsilon$ and thus the revenue is
$1+\epsilon$.
\end{itemize}

For this solution, we define the $q$ variables as follows.

\begin{itemize}
\item For buyers $i \in [k]$, we set {$q_{i,1} = (1-\theta)/k$} and $q_{i,1+\epsilon} = 1 - q_{i,1}$.
\item For buyers $i = k+1, k+2$, we set $q_{i,0} = \theta$ and $q_{i,1+\epsilon} =  1 - q_{i,1}$.
\item For buyer $k+3$, {we  set $q_{k+3,1+\epsilon} = 1$.}
\end{itemize}

It is easy to see that this solution is feasible and that it is the optimal
solution to Problem \eqref{LP}. {This is so because for any auction, any profile that has a positive weight yield the maximum revenue for that auction.}  Note that for simplicity in the formulation of revenue, we can remove the terms that are a factor of $\epsilon$ since they can be arbitrary small and are negligible. 
We argue that the rounding procedure produces a $1/(c+1)$
approximation. First notice that the vector of zero reserves obtains revenue of
$1/(c+1)$.

Now, we compute the expected revenue from rounding. After rounding, the reserve of
any buyer $i\in [k]$ is either $1$ or $1+\epsilon$, the reserve of  buyers
$k+1$ and $k+2$ is either zero or $1+\epsilon$, and reserve of buyer $k+3$
is always $1+\epsilon$. Thus, {by letting $\epsilon$ go to zero,}  the expected revenue from rounding is given by
$$ \frac{1}{c+1} \left[ 1- \left(1-\frac{1-\theta}{k} \right)^k
\cdot (1-\theta^2) \right]+\frac{c}{c+1} \,,$$
where the first term is the revenue of first auction and the second term, i.e.,
$\frac{c}{c+1}$, is the revenue of the second  auction.\footnote{{We do not
include the revenue of the third auction because we would like to take
$\epsilon$ to zero and in that case, the revenue of the third auction approaches
zero.}} {To see why the latter holds note that in the first auction, we
always  get a revenue of one unless none of the first $k$ buyers have  a reserve
of one and neither buyers $k+1$ nor buyer $k+2$ have a reserve of zero. }  As $k
\rightarrow \infty$,  the expected revenue after rounding becomes:
$$ \frac{1}{c+1} \left[ 1-e^{\theta-1}
\cdot (1-\theta^2) \right]+\frac{c}{c+1}  =    \frac{1-c}{c+1}+\frac{c}{c+1} = \frac{1}{c+1}\,,   $$
{where the first equation holds because $c = (1-\theta^2) e^{\theta-1}$. The
above equation is the desired result because the optimal revenue is at most $1$
and  $1/(c+1) = \apxf$. The latter follows from $c = (1-\theta^2)
e^{\theta-1}$ and $\theta =\sqrt{2}$. } \Halmos\end{proof}

\section{Integrality Gap}\label{sec:gap} 
\revision{In this section, we give an upper-bound of $0.828$ for the integrality gap of the LP. This implies that any rounding procedure for our LP formulation will
obtain at most $0.828$ fraction of the optimal value of the LP. In particular, we show Theorem \ref{thm:integrality_gap}, which we restate here for convenience.}

{
\begin{manualtheorem}{2}[Integrality Gap of \ref{LP}] 
There exists a data-set of bids $\{\beta_{\a, \bu}\}_{\a\in \A, \bu\in \B}$ for
  which the integrality gap of \ref{LP} is at least $ 2(\sqrt{2}-1) \approx  0.828$. That is, 
  $$\esp^{\star} \leq 2(\sqrt{2}-1) \cdot {\sf LP}^\star,$$
  where  ${\sf LP}^\star$ is the optimal fraction solution of the \ref{LP} and
  $\esp^{\star}$ is its optimal integral solution.
\end{manualtheorem}
}

\begin{proof}
Given $n$ buyers, an integer $k > 0$, $\delta = 1/k$ and a constant
$\lambda \in (0,1)$ to be determined later, consider an instance built as
follows:
\begin{itemize}
  \item {\textbf{Type one Auctions:}} For any buyer, $\bu \in [n]$,  we have an auction in which all
  the bids are zero except the bid of buyer $\bu$. Precisely, buyer $\bu$ has a
    bid of {$\lambda n$}.
  \item {\textbf{Type two Auctions:}} For any pair of buyers $\bu_1$ and $\bu_2$, there are $k$ copies
    of an auction in which $\bu_1$ and $\bu_2$ bid $\delta =  1/k$ and the rest
of the buyers bid 0. {We assume that $\lambda n > \delta$. }
\end{itemize}

For this instance, consider the fractional solution that assigns $s_{a, p} =
1/2$ for any auction $a$ of type two and profiles $ (\bu_1, \bu_0, \delta, 0)$
and $(\bu_2, \bu_0, \delta, 0)$. {For the rest of the valid profiles of
auction $\a$, we set $s_{a, p} $ to zero.} Note that $\bu_1$ and $\bu_2$ are the
buyers with nonzero bids in auction $\a$ and $\bu_0$ is a dummy  buyer.
Moreover, for any auction $\a$ of type one, in which buyer $\bu$ has a nonzero
bid, we have $s_{a, p}= 1/2$ for profile $p=(\bu, \bu_0, {\lambda n}, 0)$. {For
the rest of the valid profiles of this auction, we set $s_{a, p} $ to zero. } In
this solution for any buyer $\bu$, we have $q_{\bu, \delta} = 1/2$ and $q_{\bu,
\lambda n} = 1/2$. One can simply verify that this solution satisfies all the
constraints of the LP and as a result, it is a valid fractional solution. The
optimal value of the  LP  is therefore bounded by:

  {
    $$ {\sf LP}^\star \geq \sum_\a \Rev_\a(\mathbf{s}) = n \cdot \frac{\lambda}{2} n + {n\choose2} \cdot
  k \cdot \delta = \frac{1+\lambda}{2} \cdot n^2 + o(n^2)\,,$$
  where the first term corresponds to the revenue from auctions of type one and
  the second term corresponds to the revenue of auctions of type two. To bound $\esp^{\star}$,
we note that in the optimal solution of Problem (\ref{eq:opt}), the reserve of each buyer is either
  $\delta$ or $\lambda n$. Given that the buyers are symmetric, the value of the
  optimal solution depends only on the number of buyers with each reserve. Let $t$ be the number of buyers with reserve $\lambda n$. Then, we
  can write:
  {$$ \esp^{\star} = \max_{0 \leq t \leq n} \left[ t \cdot \lambda n +(n-t)\cdot \delta+ {n \choose 2} - {t \choose 2}\right].$$}
 {By taking $\delta\rightarrow 0$, we obtain,}
   $${ \esp^{\star} = \max_{0 \leq t \leq n} \left[ t \cdot \lambda n+ {n \choose 2} -
  {t \choose 2}\right].}$$
  Since the term inside the maximum is a quadratic function of $t$, the optimal
  integral solution should be $t = \lambda n + o(n)$. This is so because the optimal integral
  solution $t$ deviates from the optimal fractional solution (which is {$\lambda
  n +1/2$}) by at most $1$. Substituting that in the expression of
  $\esp^{\star}$, we get
  $$  \esp^{\star} = \frac{1+\lambda^2}{2} \cdot n^2 + o(n^2)\,. $$
  Taking $n \rightarrow \infty$, we get
  $$\frac{\esp^{\star}  }{ {\sf LP}^\star } \leq \frac{(1 + \lambda^2) n^2 +
  o(n^2)}{(1 + \lambda) n^2 + o(n^2)} \rightarrow \frac{1 +
  \lambda^2}{1+\lambda}. $$
  We can choose the parameter $\lambda = \sqrt{2}-1$ to minimize the above
  expression, which leads to a ratio of $2(\sqrt{2}-1) \approx 0.828$.
}
\end{proof}

\section{{Numerical Studies}}
\label{sec:numerical}
{In this section, we evaluate our Pro-LPR algorithm on synthetic data-sets. As a benchmark, we use the optimal value of Problem \eqref{LP} and 
the greedy algorithm proposed by \cite{RW16}. (Recall that 
  optimal value of Problem \eqref{LP} provides an upper bound on the revenue obtained by any vector of reserve prices.)
We assess the performance of our algorithm in three settings, 
 where in the first (respectively third) setting, the bids  are negatively (respectively positively) correlated across buyers. In the second setting,  bids are independent of each other. }
 
 {\textbf{Synthetic Data-sets:} We assume that the number of buyers $n=2$. By setting $n=2$, we provide a  fair comparison between Pro-LP and greedy algorithms as  with $n=2$, the Pro-LP algorithm, similar to the greedy algorithm, only uses the highest and second-highest bids in each auction.
 The bids in each auction $\a$ are drawn from log-normal distributions. Note that  log-normal distributions are proved to be  a good fit for the advertisers' valuations in online advertising markets; see, for example,  \cite{edelman2006optimal, edelman2007internet, xiao2009optimal}, and \cite{balseiro2014yield}.
 Let $\sigma =0.1$ and define matrix $\Sigma$ as follows 
 \[\Sigma =\sigma^2 \left( \begin{array}{cc}
1 & w  \\
w & 1   \end{array} \right)\,.\] 
To generate  bids in each auction $\a$, we first generate normal random variables, denoted by $Z$, with mean $(0, \mu)$ and variance $\Sigma$. The bids are then obtained by computing $e^Z$. Here, we focus on three values of $w$, $-0.2, 0$, and $0.2$. When $w=0$, bids are independent of each other and when $w>0$ (respectively $w<0$), bids are positively (respectively negatively) correlated. More precisely,
the correlation coefficient between bids of the two buyers  is equal to $\frac{(\exp(w\sigma^2)-1)}{(\exp(\sigma^2)-1)}$.  We now comment on the parameter $\mu$. 
We draw $\mu$ randomly from a uniform distribution in the range of $[0,1]$, where each value of $\mu$ represents one problem instance. We consider $50$ problem instances for each value of $w\in \{-0.2, 0, 0.2\}$. Our generated data-set consists of $100$ auctions for each value of $\mu$ and $w$. (Note that bids are independent across auctions.) After generating the data-set, we consider the $30$ equally spaced values between zero and the maximum bid in the  data-set as candidate reserve prices. We use the same set of candidate reserve prices in Problem \eqref{LP} and the greedy algorithm of  \cite{RW16}.}

{\textbf{Evaluation:} We evaluate the revenue of Pro-LPR and greedy algorithms
  in two ways. First, we compare their revenue to the optimal value of Problem \eqref{LP} on the original data-set that we constructed to calculate the reserve prices. We refer to this data-set as the \emph{training data-set}.  To evaluate our algorithm on the training data-set, we draw $200$ vector of reserve prices from distribution $q^\star$, and choose the best of these reserve prices and zero reserve prices. 
Here, $q^\star$ is the optimal solution to the problem (\ref{LP}) when the input is the training data-set. We then report the average revenue of these $200$ random draws as the revenue of our algorithm on the training data-set. Figure \ref{fig:compare_to_opt} shows the distribution of  normalized revenue of the Pro-LPR and 
	 greedy algorithms on the training data-sets when bids are
	independent ($w=0$), positively correlated ($w=0.2$), and negatively correlated ($w=-0.2$). Here, the normalization is done with respect to the upper bound, i.e., the optimal value of Problem \eqref{LP}. 
	We observe that in the worst problem instance, our algorithm obtains a $0.98$ fraction of the upper bound, whereas the greedy algorithm in the worst case only obtains approximately an  $0.82$ fraction of the upper bound. Furthermore, while our algorithm matches the upper bound in at least half of the problem instances, the greedy algorithm does not meet the upper bound even in its best problem instance. }

\begin{figure}
	\centering
	\includegraphics[width=4.5in]{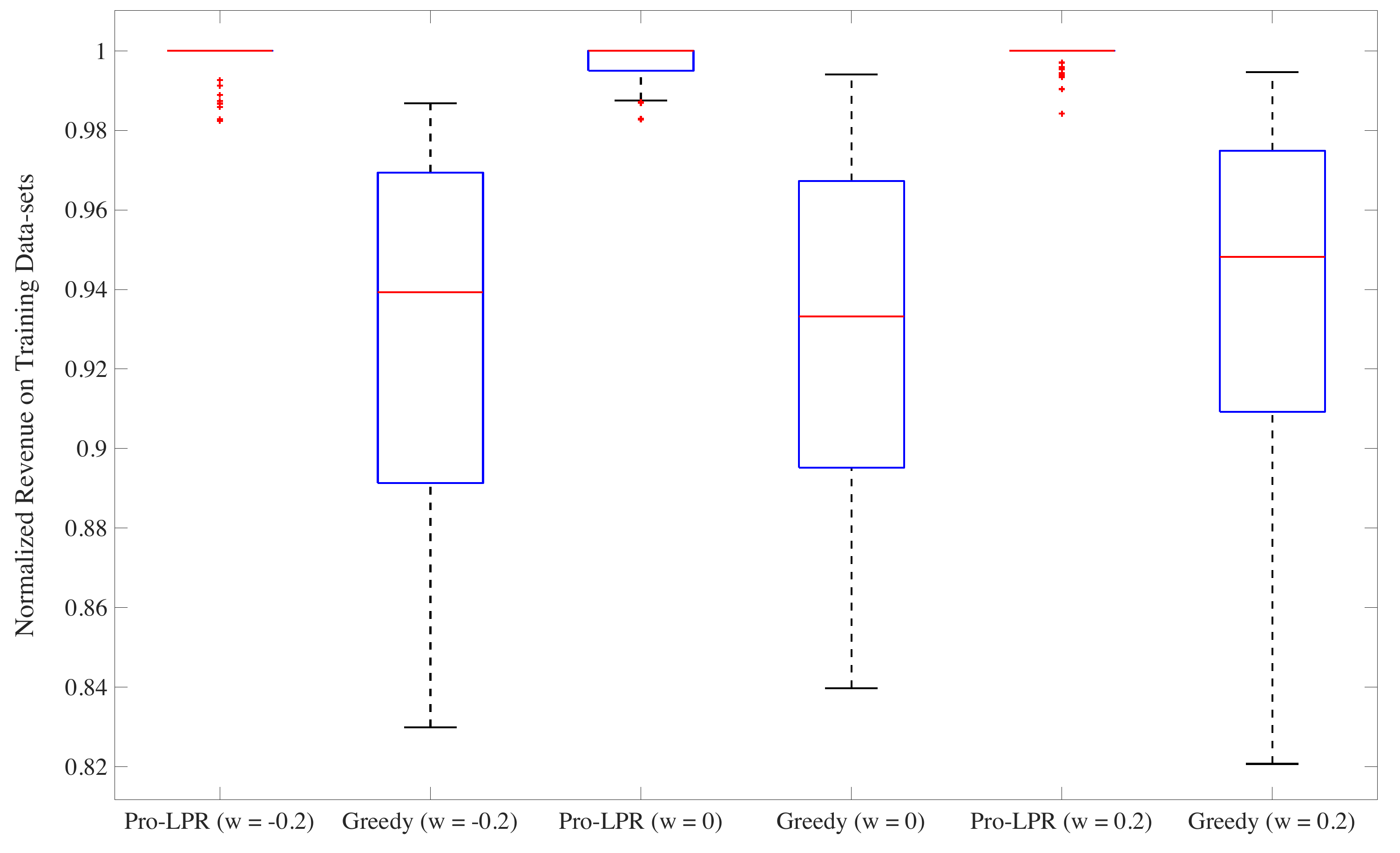}
	\caption{
	The normalized revenue of the Pro-LPR algorithm and greedy algorithm of \cite{RW16} on the training data-sets when bids are
	independent ($w=0$), positively correlated ($w=0.2$), and negatively correlated ($w=-0.2$). Here, the normalization is done with respect to the optimal value of Problem \eqref{LP}, which is an upper bound on revenue of any vector of reserve prices. \label{fig:compare_to_opt}} 
\end{figure}

{Second, we evaluate the revenue of our algorithm on \emph{test data-sets}. Specifically, for each value of $(\mu, w)$, we construct $100$ test data-sets where each test data-set consists of $100$ auctions whose bid distributions are the same as the training data-set with the same value of $(\mu, w)$. On the test data-sets, the upper bound is irrelevant. Thus,  we only report the revenue gain of our algorithm relative to the greedy algorithm. To assess our algorithm, for each value of $(\mu, w)$, we use one of the $200$ random reserve prices that obtained  the highest revenue on the training data-set. 
Figure \ref{fig:revenue_gain} presents three boxplots, where each boxplot shows the distribution of the  revenue gain (in percentage) of our algorithm (relative to the greedy algorithm) on test data-sets  for $w\in\{-0.2, 0, 0.2\}$. Note that in each boxplot we have $50\times 100= 5000$ data-points, where each data-point corresponds to one value of $\mu$ and one test data-set. Recall that for every value of $(\mu, w)$, we construct $100$ test data-sets and consider $50$ (random) values for $\mu$.  We observe that the median of the  revenue gain  is between $6\%$ and $7\%$. Furthermore, the $25$-th percentile in all the boxplots is between $2\%-3\%$. This means that  in $75\%$ of the problem instances,  our algorithm beats the greedy algorithm by at least $2\%-3\%$.    }

\begin{figure}
	\centering
	\includegraphics[width=4.5in]{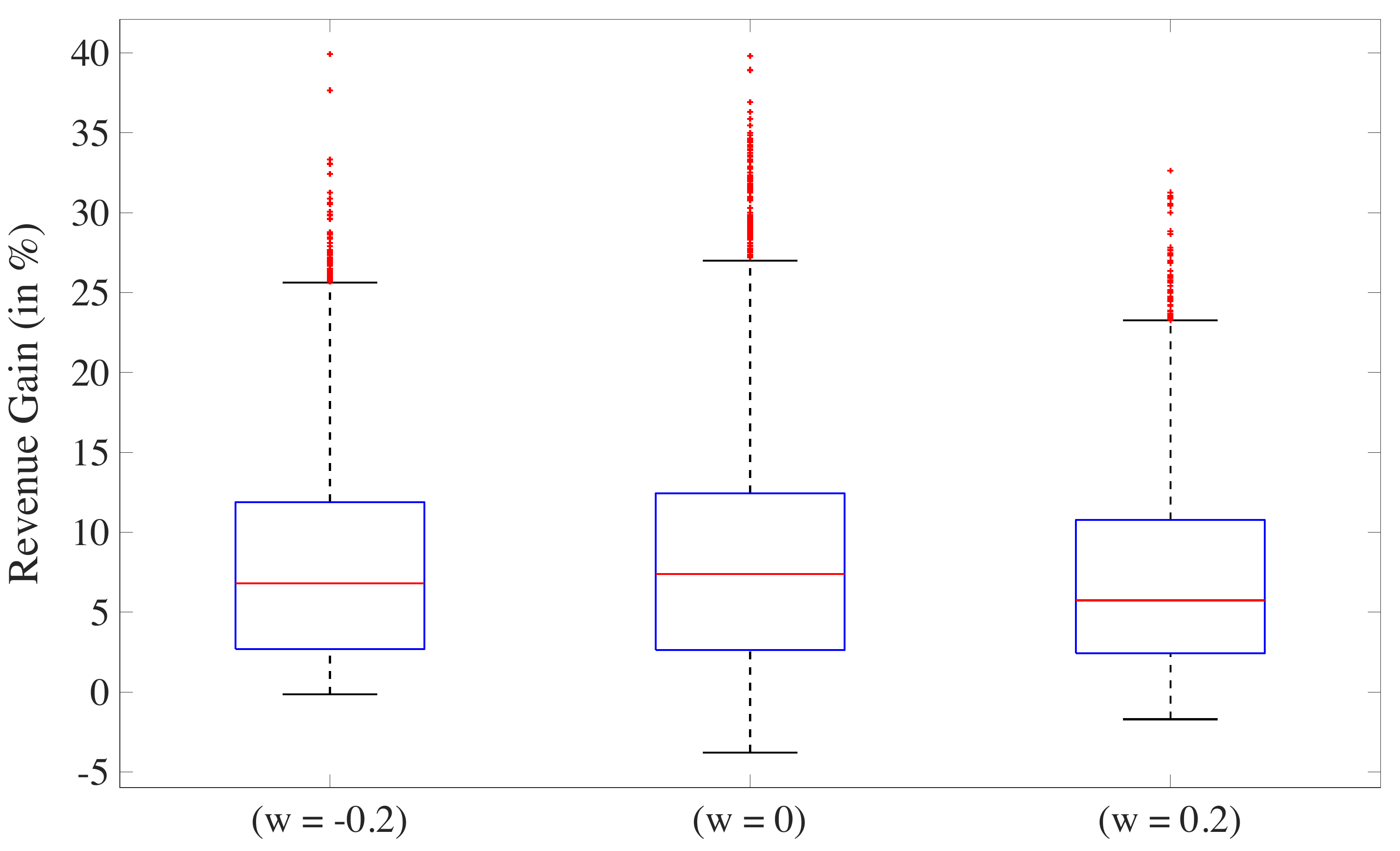}
	\caption{\label{fig:revenue_gain}
	The revenue gain of Pro-LPR algorithm (relative to the greedy algorithm of \cite{RW16})  on test data-sets when bids are independent ($w=0$), positively correlated ($w=0.2$), and negatively correlated ($w=-0.2$). }
\end{figure}

\section{Conclusion}\label{sec:conclusion}
In this paper, we take a data-driven approach to optimize personalized reserve  prices  in  eager second price auctions. We design  a polynomial time LP-based algorithm to optimize reserve prices on a given data-set of submitted bids and show that our algorithm  obtains more than \apxf~ fraction of the optimal revenue.  {Our algorithm, which  takes advantage of all the submitted bids to devise an effective reserve prices, highlights the importance of  deviating from greedy policies to optimize reserve prices.} {Furthermore, our theoretical results and numerical studies confirm that our algorithm  performs well  when bids are correlated across buyers or independent of each other. Nonetheless, it is an exciting future research direction to explore if one can design a data-driven algorithm with a better approximation factor when bids are independent of each other. }

\revision{Another exciting future research direction is to explore how to ``transform" our LP-based algorithm to an online learning algorithm with a sublinear approximate regret. 
For the greedy algorithm of \cite{RW16}, such transformation has shown to be possible using the Follow-the-Perturbed-Leader algorithm (\cite{RW16}) and Blackwell Approachability (\cite{niazadeh2020online}).  However, the greedy algorithm can lose up to a $1/2$ fraction of the optimal revenue, as opposed to a $1-\apxf$ fraction of the optimal revenue that our algorithm can lose. This makes transforming our algorithm to its online counterpart an interesting future research direction. Finally, it is worth exploring if/how correlated rounding techniques can improve our approximation factor. In particular, one can explore if it is possible to close the gap between our approximation factor of  $\apxf$ and the integrality gap of  $0.828$.}

{To sum up, we believe that our data-driven approach, as well as our LP-based algorithm can also be applied to a wider class of problems with revenue objective, and we hope the framework in this paper serves as a starting point for designing other data-driven algorithms.}

\bibliographystyle{alpha}
\bibliography{ref}
	
\end{document}